%% file: als.tex
\documentclass[preprint]{elsarticle}

\usepackage{amsfonts}
\usepackage{amsmath}
\usepackage{amsthm}
\usepackage{amssymb}
\usepackage{mathrsfs}
\usepackage{verbatim}
\usepackage{pgfplotstable}
\usepackage{bm}

\usepackage{url}
\usepackage{hyperref}

\input{mcomsymb}

\input{mcomdefs}

\input{mengthm}

\input{malghypdefs}

\begin{document}
\author[g,g2]{Konstantin Usevich\corref{cor}}\ead{konstantin.usevich@gipsa-lab.grenoble-inp.fr}
\author[b]{Ivan Markovsky}\ead{ivan.markovsky@vub.ac.be}
\cortext[cor]{Corresponding author}
\address[g]{Univ. Grenoble Alpes, GIPSA-Lab, F-38000 Grenoble, France}
\address[g2]{CNRS, GIPSA-Lab, F-38000 Grenoble, France}
\address[b]{Department ELEC, Vrije Universiteit Brussel, Pleinlaan 2,  B-1050 Brussels, Belgium}
\title{Adjusted least squares fitting\\ of algebraic hypersurfaces}
\date{}

\begin{abstract} 
We consider the problem of fitting a set of points in Euclidean space by
an algebraic hypersurface. 
We assume that points on a true hypersurface, described by a polynomial equation, are corrupted by zero mean independent Gaussian noise, 
and we estimate the coefficients of the true polynomial equation. 
The adjusted least squares estimator accounts for the bias present in the ordinary least squares estimator. 
The adjusted least squares estimator is based on constructing a quasi-Hankel matrix, which is a bias-corrected matrix of moments. 
For the case of unknown noise variance, the estimator is defined as a solution of a polynomial eigenvalue problem.
In this paper, we present new results on invariance properties of the adjusted least squares estimator and an improved algorithm for computing the estimator for an arbitrary set of monomials in the polynomial equation.
\end{abstract}


\maketitle

\input{intro.tex}

\input{maindefs.tex}

\input{theory.tex}
\input{invariance.tex}

\input{applications.tex}

\section*{Acknowledgements}
This work was supported by European Research Council under the European Union's Seventh Framework Programme (FP7/2007-2013) / ERC Grant Agreement No. 258581 ``Structured low-rank approximation: Theory, algorithms, and applications'' and Grant Agreement No. 320594 DECODA project.

\section*{References}
\bibliography{als}
\end{document}

%% file: mcomsymb.tex
\def\scrD{\mathscr{D}}

\def\scrH{\mathscr{H}}

\def\calA{\mathcal{A}}

\def\calD{\mathcal{D}}

\def\calH{\mathcal{H}}
\def\calI{\mathcal{I}}
\def\calJ{\mathcal{J}}

\def\calN{\mathcal{N}}

\def\bfE{\mathbf{E}}

\def\bbR{\mathbb{R}}

\def\rmB{\mathrm{B}}
\def\rmC{\mathrm{C}}

\def\rmH{\mathrm{H}}

\def\rmM{\mathrm{M}}

\def\rmB{\mathrm{B}}
\def\rmC{\mathrm{C}}

\def\rmH{\mathrm{H}}

\def\rmM{\mathrm{M}}

%% file: mcomdefs.tex
\newcommand{\defeq}{\mathrel{:=}}

\newcommand{\sto}{\mathop{\text{\;subject\ to\;}}}
\newcommand{\diag}{\mathop{\mathrm{diag}}} 

\newcommand{\bmx}{\begin{bmatrix}}
\newcommand{\emx}{\end{bmatrix}}
\newcommand{\bsm}{\left[\begin{smallmatrix}}
\newcommand{\esm}{\end{smallmatrix}\right]}

\def\bbZp{\mathbb{Z}_{+}}


%% file: mengthm.tex
\newtheorem{theorem}{Theorem} 
 
\newtheorem{corollary}{Corollary} 
\newtheorem{definition}{Definition} 
\newtheorem{lemma}{Lemma} 
 
\newtheorem{note}{Note} 
 
\newtheorem{example}{Example} 
 
\newtheorem{assumption}{Assumption}

%% file: malghypdefs.tex
\def\normal{\calN}
\def\dpoint{d}
\def\dpoints{\mathscr{D}}

\def\mcoef{\rmM}
\def\mCoef{\rmM}
\def\mCoefBas#1{\rmM^{(b,s,#1)}}
\def\mCoefSig{\rmM^{(\sigma,s)}}
\def\mcoefSig{\rmM^{(\sigma,s)}}
\def\herm#1#2#3{h^{(#2)}(#1,#3)}

\def\sign{\mathop{\mathrm{sign}}}
\def\argmin{\mathop{\mathrm{argmin}}}

\def\vandmat#1#2{{\Phi}_{#1}(#2)}

\def\psimatcor#1#2#3{\Psi_{\text{als},#3}(#2)}
\def\boxset#1{\blacksquare^{(#1)}}
\def\trgset#1#2{\blacktriangle^{(#1,#2)}}
\def\degset#1#2{\triangle^{(#1,#2)}}

\def\nvar{{\tt q}}
\def\nmind{m}
\def\npoints{N}
\def\typdeg{\ell}

\def\sthat{\mathop{|}}

%% file: intro.tex
\section{Introduction}
An algebraic hypersurface is the set of points $d \in \bbR^{\nvar}$ that are the solutions of an \emph{implicit polynomial equation}
\begin{equation}\label{eq:alg_hyp_equation}
R_{\theta}(d) = 0.
\end{equation}
In \eqref{eq:alg_hyp_equation}, $R_{\theta}(d)$ is a multivariate polynomial with coefficients
$\theta = \bmx\theta_1 & \cdots & \theta_{\nmind} \emx^{\top}$ 
\begin{equation}\label{eq:alg_hyp}
R_{\theta}(d) := \theta^{\top} \phi(d) = \sum\limits_{j=1}^{\nmind} \theta_{j} \phi_j(d),
\end{equation}
where $\phi(d)$ is the vector of linearly independent basis polynomials
\begin{equation}\label{eq:monomial_vec}
\phi(d) = \bmx \phi_1(d) & \cdots & \phi_{\nmind} (d) \emx^{\top}, \quad d \in \bbR^{\nvar}.
\end{equation}
The algebraic hypersurface fitting problem is to fit a given set of points
\[
\dpoints = \{\dpoint^{(1)}, \ldots, \dpoint^{(\npoints)}\} \subset \bbR^{\nvar}, 
\]
in the best way by an algebraic hypersurface of the form \eqref{eq:alg_hyp_equation}, where the vector of basis monomials is given and fixed. The notion of ``best'' is determined by  a chosen goodness-of-fit measure.

Fitting two-dimensional data by conic sections ($\nvar=2$) is the most common case of algebraic hypersurface fitting, with numerous applications in robotics, medical imaging, archaeology, etc., see \cite{Chernov10-Circular} for an overview. Fitting algebraic hypersurfaces of higher degrees and dimensions is needed in computer graphics \cite{Pratt87SCG-Direct}, computer vision \cite{Taubin91IToPAMI-Estimation}, and symbolic-numeric computations \cite{Sauer07NA-Approximate}, \cite[\S 5]{Emiris.etal13TCS-Implicitization}.  The  problem also appears in advanced methods of multivariate data analysis such as subspace clustering \cite{Vidal11ISPM-Subspace} and non-linear system identification \cite{Vajk.Hethessy03A-Identification}, see \cite{kpca} for an overview. Algebraic hypersurface fitting received considerable attention in linear algebra community starting from \cite{Gander.etal94BNM-Least}. Recently, it has been shown to be an instance of nonlinearly structured low-rank approximation \cite[Ch. 6]{Markovsky12-Low}. 

The most widespread fits are \emph{geometric} and \emph{algebraic} fit (see, for example, \cite{Gander.etal94BNM-Least} and \cite{Markovsky.etal04NM-Consistent}). The geometric fit minimizes total distance from $\dpoints$ to a hypersurface defined by \eqref{eq:alg_hyp_equation}. Although this is a natural goodness-of-fit measure, the resulting nonlinear optimization problem is difficult. 

A computationally cheap alternative to geometric fit is the algebraic fit, which minimizes the sum of squared residuals of the implicit equation \eqref{eq:alg_hyp_equation}
\begin{equation}\label{eq:ols_cost_def}
Q_{\text{ols}}(\theta,\dpoints) := \sum_{k=1}^{\npoints} \left(R_{\theta}(\dpoint^{(k)})\right)^2.
\end{equation}
More precisely, the algebraic fit is defined as the solution of 
\begin{equation}\label{eq:ols_opt_problem}
\begin{split}
{\widehat{\theta}}_{\text{ols}} := & \argmin_{\theta \in \bbR^{m}} Q_{\text{ols}}(\theta,\dpoints)  \, \\
                          & \sto \|\theta \| = 1,
\end{split}
\end{equation}
where the normalization $\|\theta \| = 1$ is  needed, since multiplication of $\theta$ by a nonzero constant does not change the hypersurface. 
In statistical literature  \cite{Kukush.etal04CSDA-Consistent,Shklyar.etal07JoMA-conic},
 the algebraic fit is known under the name \emph{ordinary least squares} (OLS), which explains the notation ${\widehat{\theta}}_{\text{ols}}$.
 
If $\|\cdot\|$ is a weighted $2$-norm, then $\widehat{\theta}_{\text{ols}}$ can be found as an eigenvector of a matrix $\Psi(\dpoints)$ constructed from data $\scrD$ (equivalently, a singular vector of the multivariate Vandermonde matrix, see Section~\ref{sec:preliminaries} for more details). The simplicity of the algebraic fit makes it  a method of choice in many applications \cite{Vidal11ISPM-Subspace,Sauer07NA-Approximate}. Also, $\widehat{\theta}_{\text{ols}}$ is often used as an initial guess for local optimization in finding the geometric fit.

Despite of its popularity, the algebraic fit has many deficiencies. A desirable property of a hypersurface fitting is invariance with respect to  all similarity transformations (compositions of translation, rotation and uniform scaling) \cite{Gander.etal94BNM-Least}. The algebraic fit is, in general, not invariant to the similarity transformations; as a result, the 
locally optimal geometric fit initialized with the algebraic fit is also not invariant to these transformations. 

Moreover, the algebraic fit fails in the statistical estimation framework. Assume that $\dpoints$ are generated as 
\begin{equation}\label{eq:eiv}
\dpoint^{(j)} = \overline{\dpoint}^{(j)} + \widetilde{\dpoint}^{(j)},\quad j=1,\ldots,N,
\end{equation}
where $\overline{\dpoint}^{(j)}$ lie on a true hypersurface $R_{\overline{\theta}} (d)= 0$, and the errors are Gaussian
\begin{equation}\label{eq:gaussian_errors}
\widetilde{\dpoint}^{(1)}, \ldots, \widetilde{\dpoint}^{(N)} \sim \normal(0, \Sigma), \quad \{ \widetilde{\dpoint}^{(1)},\ldots, \widetilde{\dpoint}^{(N)}\} \text{ are independent}.
\end{equation}
In this case, it is known \cite{Kukush.etal04CSDA-Consistent} that the ${\widehat{\theta}}_{\text{ols}}$ is not a consistent estimator. 
Informally speaking, if the estimator is inconsistent, the accuracy of estimation of the true parameter $\overline{\theta}$ may not be improved by increasing the number of observed noisy data points $N$ \cite[Ch. 17]{Kendall.Stuart77-advanced}. The algebraic fit is inconsistent due to bias (systematic error) of ${\widehat{\theta}}_{\text{ols}}$ caused by the nonlinear structure of $\Psi(\scrD)$. Besides, geometric fit is also inconsistent, but has a smaller asymptotic bias \cite{Kukush.etal04CSDA-Consistent}.

Many heuristic methods were proposed to overcome the aforementioned problems of algebraic and geometric fit, see a recent book of Chernov \cite{Chernov10-Circular} for an overview (the methods are mainly proposed for ellipsoid fitting). 
In this paper, we focus on \emph{adjusted least squares} fitting, which combines advantages of algebraic and geometric fit, such as low computational complexity, invariance under similarity transformations, and good statistical properties. 

Assume that the data points are generated according to \eqref{eq:eiv} and \eqref{eq:gaussian_errors}, with $\Sigma = \sigma^2 \Sigma_0$. 
The first version of ALS estimator $\widehat{\theta}_{\text{als},\sigma}$ is constructed as eigenvector of the adjusted (bias-corrected) matrix $\Psi$, denoted by $\psimatcor{}{\dpoints}{\sigma}$. The second version of the ALS estimator, for unknown $\sigma^2$, is defined as an eigenvector of the matrix $\psimatcor{}{\dpoints}{\widehat{\sigma}}$, where $\widehat{\sigma} \ge 0$ is a solution of
\begin{equation}\label{eq:als2_equation}
\lambda_{min}(\psimatcor{\calA}{\dpoints}{\widehat{\sigma}}) = 0,
\end{equation}
see Section~\ref{sec:preliminaries} for more details.

The estimators $\widehat{\theta}_{\text{als},\sigma}$ and $\widehat{\theta}_{\text{als}}$ were initially proposed by Kukush, Markovsky and Van Huffel in \cite{Kukush.etal04CSDA-Consistent} for quadratic hypersurfaces and $\Sigma_0 = I$. Some authors, like Chernov \cite{Chernov10-Circular}, use the abbreviation KMvH (from the first letters of the surnames of the authors of \cite{Kukush.etal04CSDA-Consistent}). In this paper, we stick to the name ``adjusted least squares'' (ALS), but we prefer to use the expression ``ALS fitting'' (instead of ``ALS estimation'') in order to emphasize that the ALS fitting can be used beyond the restrictive probabilistic model of errors \eqref{eq:gaussian_errors}.

In Fig.~\ref{fig:conic}, we illustrate the fitting problem for conic sections, i.e. $\nvar = 2$ and the vector $\phi(d)$ given by 
\begin{equation}\label{eq:conic_monomials}
\phi(d) = \bmx d^2_1     & d_1d_2    & d_2^2     & d_1       & d_2       & 1      \emx^{\top}.
\end{equation}
The true points lie on a parabola, and they are perturbed by Gaussian noise. For this example, the ALS fitting is more accurate than the algebraic fit.

\begin{figure}[ht]%
\centering%
\includegraphics[width=0.5\textwidth]{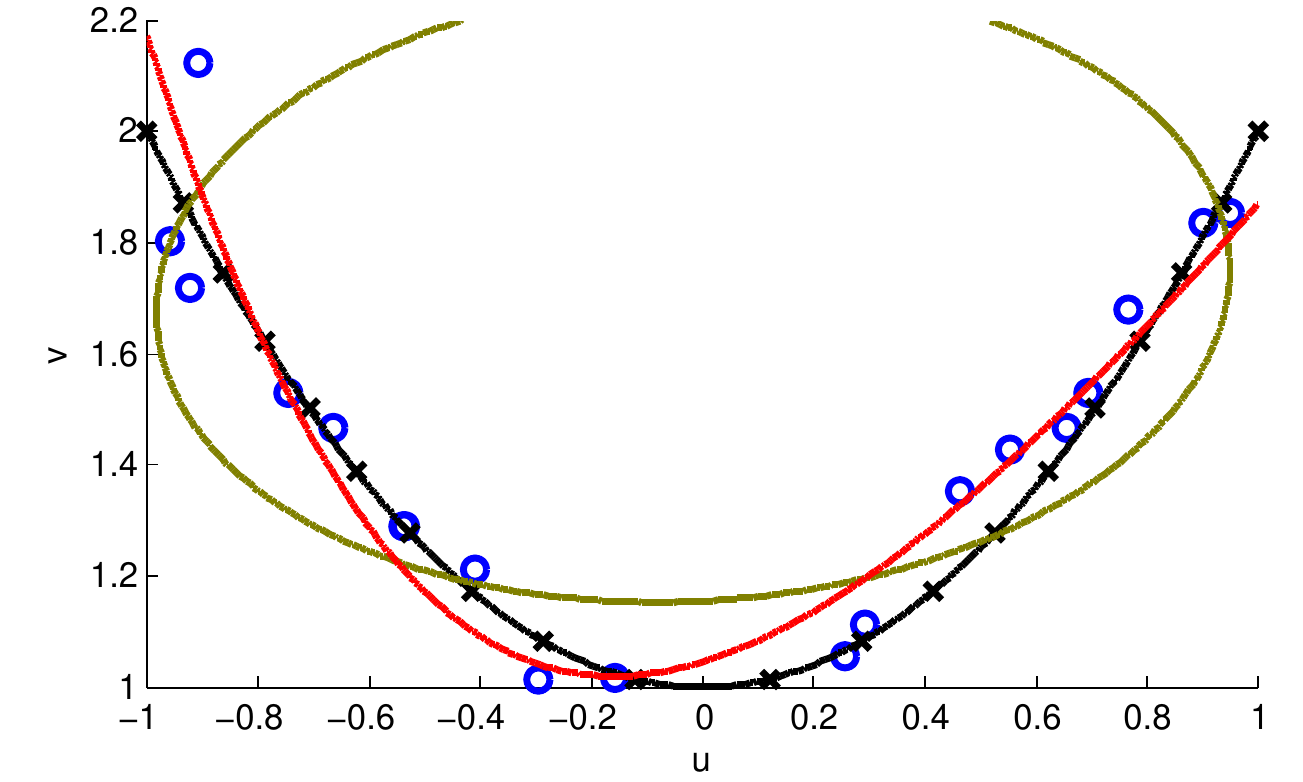}%
\caption{Fitting a parabola: black curve with crosses ---  true curve $R_{\overline{\theta}}(d) = 0$ and true points $\overline{d}^{(j)}$; blue circles --- perturbed points ${d}^{(j)}$; red curve --- approximating curve for $\widehat{\theta}_{\text{als}}$; olive curve --- approximating curve for $\widehat{\theta}_{\text{ols}}$. }%
\label{fig:conic}
\end{figure}
 
For algebraic hypersurfaces of higher degrees, the adjusted least squares fitting was independently proposed by Markovsky \cite[Ch. 6]{Markovsky12-Low} and Shklyar \cite{Shklyar09PhD-Consistency} (published in Ukrainian). In  \cite{Shklyar09PhD-Consistency}, for a general $\Sigma_0$ and for a class of polynomial vectors $\phi(d)$, existence, uniqueness and consistency of $\widehat{\theta}_{\text{als}}$ was proved. However, the construction in \cite{Shklyar09PhD-Consistency} is abstract, and an algorithm for computation of $\widehat{\theta}_{\text{als}}$ was not given. In \cite[Ch. 6]{Markovsky12-Low}, Markovsky showed that \eqref{eq:als2_equation} is equivalent to a polynomial eigenvalue problem, 
and described an algorithm for computing $\widehat{\theta}_{\text{als}}$. (Previously, for quadratic hypersurfaces, the equation \eqref{eq:als2_equation} was solved by bisection \cite{Kukush.etal04CSDA-Consistent,Shklyar.etal07JoMA-conic}.) However, in \cite[Ch. 6]{Markovsky12-Low} only a special case was considered ($\Sigma_0 = I$ and $\phi(d)$ equal to a vector of all monomials of total degree $\le r$).
Also, the coefficients of the matrix polynomial were computed by a recursive algorithm, and the structure of matrix coefficients of the matrix polynomial was not studied.

\paragraph{Contribution of the paper}
In this paper, we consider the general case of $\Sigma_0$ and general set of basis polynomials $\phi(d)$. 
We show that in the general case $\widehat{\theta}_{\text{als}}$ can  also be computed as a solution of a polynomial eigenvalue problem.
Compared with \cite[Ch. 6]{Markovsky12-Low}, we derive the explicit form of the matrix polynomial coefficients. 
If $\phi(d)$ is an arbitrary vector of monomials, we show that the coefficients of the matrix polynomial are quasi-Hankel matrices constructed from the shifts of the array of moments of data.
This simplifies the computation of $\widehat{\theta}_{\text{als}, \sigma}$ and $\widehat{\theta}_{\text{als}}$, and gives an alternative condition for existence of $\widehat{\theta}_{\text{als}}$.
Finally, we derive conditions for rotational/translational/scaling invariance of $\widehat{\theta}_{\text{ols}}$, $\widehat{\theta}_{\text{als}, \sigma}$ and $\widehat{\theta}_{\text{als}}$. In particular, we show that it is important to use the Bombieri norm in order to achieve rotational invariance of $\widehat{\theta}_{\text{ols}}$ and $\widehat{\theta}_{\text{als}, \sigma}$.
We provide numerical results that support the theoretical results of the paper.
Moreover, the numerical results suggest that $\widehat{\theta}_{\text{als}}$ can be used beyond the model \eqref{eq:gaussian_errors} of Gaussian errors and existing consistency results; thus ALS fitting can be used as a general-purpose hypersurface fitting method.
The software implementing the ALS fitting methods, together with reproducible examples, is available at \url{http://github.com/slra/als-fit}.

The paper is organized as follows. In Section~\ref{sec:preliminaries}, we review the existing results. 
In Section~\ref{sec:sets}, we remind the details of construction of $\widehat{\theta}_{\text{ols}}$. Then we give a definition of  $\widehat{\theta}_{\text{als}, \sigma}$ and $\widehat{\theta}_{\text{als}}$ through the deconvolutions (in spirit of \cite{Shklyar09PhD-Consistency}). We provide a brief summary of the main results of \cite{Shklyar09PhD-Consistency} (available only in Ukrainian). Then, we show how the deconvolution can be performed with Hermite polynomials by recalling the construction of \cite[Ch. 6]{Markovsky12-Low} for the case of monomials and $\Sigma_0 = I$. In Section~\ref{sec:comp_adjusted}, we show that the estimators can be computed using quasi-Hankel matrices, and operations on the shifts of the moment array. 
As a corollary, we improve a necessary condition of \cite{Shklyar09PhD-Consistency} for existence of the ALS estimator.
In Section~\ref{sec:invariance}, the invariance properties of the estimators are studied, generalizing the results of \cite{Shklyar.etal07JoMA-conic}.
In Section~\ref{sec:numerical}, we provide numerical experiments that demonstrate the advantages of the ALS fitting.

%% file: maindefs.tex
\section{Main notation and background}\label{sec:preliminaries}
\subsection{Multidegrees and sets of multidegrees}\label{sec:sets}
A nonnegative integer vector $\alpha \in \bbZp^{\nvar}$ corresponds to the monomial $d^{\alpha}$, where 
for $d= \bmx d_1 & \cdots &d_{\nvar} \emx^{\top}$ and $\alpha= \bmx \alpha_1 & \cdots& \alpha_{\nvar} \emx^{\top}$ the power operation $d^{\alpha}$ is defined as
\[
\dpoint^{\alpha} := d_1^{\alpha_1} \cdots d_\nvar^{\alpha_\nvar}.
\]
 Therefore, we refer to nonnegative integer vectors $\alpha \in \bbZp^{\nvar}$ as  \textit{multidegrees}. For $s \le \nvar$, we define
\begin{equation}\label{eq:total_degree}
|\alpha|_s := \alpha_1 + \cdots + \alpha_{s}.
\end{equation}
We also use a shorthand $|\alpha| := |\alpha|_\nvar$, which corresponds to the \emph{total degree} of $d^{\alpha}$. $\alpha+\beta$ denotes the element-wise sum of multidegrees. For multidegrees $\alpha,\beta \in \bbZp^\nvar$, the partial order $\le$ is defined in a standard way:
\begin{equation}\label{eq:part_order}
\alpha \le \beta \iff \alpha_k \le \beta_k, \; \mbox{for}\; k=1,\ldots, \nvar.
\end{equation}
We will frequently use the following sets of multidegrees:
\begin{itemize}
\item \emph{degree-constrained set}: 
\[
\degset{\nvar}{\typdeg}:=\{\alpha \in \bbZp^{\nvar} \,|\, |\alpha| = \typdeg \};
\]
\item \emph{triangular set}: 
\[
\trgset{\nvar}{\typdeg}:= \{\alpha \in \bbZp^{\nvar} \,|\, |\alpha| \le \typdeg \};
\]
\item \emph{box set} $\boxset{\gamma}$ (for $\gamma \in \bbZp^\nvar$):
\[
\boxset{\gamma} \defeq  \{0,1,\ldots,\gamma_1\} \times \cdots \times \{0,1,\ldots,\gamma_\nvar\}.
\]
\end{itemize}
These sets are related using the following evident relations:
\begin{itemize}
\item for any $\nvar$, $\typdeg$ we have
\[
\trgset{\nvar}{\typdeg} = \degset{\nvar}{0} \cup \cdots \cup \degset{\nvar}{\typdeg};
\]
\item for $\nvar = 1$ we have
\[
\trgset{1}{\typdeg} =  {\boxset{\typdeg}} = \{0,1,\ldots,\typdeg\}, \quad \degset{1}{\typdeg} = \{\typdeg\}.
\]
\end{itemize}
The \textit{Minkowski sum} of sets of multidegrees is the set:
\[
\mathfrak{A} + \mathfrak{B} \defeq \{ \alpha + \beta \sthat \alpha \in \mathfrak{A},\; \beta \in \mathfrak{B} \}.
\]
The following examples of the Minkowski sum are evident:
\begin{itemize}
\item ${\boxset{\gamma^{(1)}}} + {\boxset{\gamma^{(2)}}} = {\boxset{\gamma^{(1)}+\gamma^{(2)}}}$;
\item ${\trgset{\nvar}{\typdeg_1}} + {\trgset{\nvar}{\typdeg_2}} = {\trgset{\nvar}{\typdeg_1+\typdeg_2}}$;
\item ${\degset{\nvar}{\typdeg_1}} + {\degset{\nvar}{\typdeg_2}} = {\degset{\nvar}{\typdeg_1+\typdeg_2}}$.
\end{itemize}
A set of multidegrees $\mathfrak{A}$ is a called \emph{lower set}, if for any $\alpha \in \mathfrak{A}$
\begin{equation}\label{eq:lower_set_def}
\beta \le \alpha \Rightarrow \beta \in \mathfrak{A}.
\end{equation}
It is easy to see that $\trgset{\nvar}{\typdeg}$ and $\boxset{\gamma}$ are lower sets, but $\degset{\nvar}{\typdeg}$ is not.

It is convenient to work with matrix representations of the sets of multidegrees.
Assume that a $\nvar \times \nmind$ nonnegative integer matrix is given
\begin{equation}\label{eq:calA}
\calA = \bmx  \alpha^{(1)} & \cdots & \alpha^{(\nmind)} \emx = 
\bmx
\alpha^{(1)}_1 & \cdots & \alpha^{(\nmind)}_1  \\
\vdots         &        & \vdots          \\
\alpha^{(1)}_\nvar & \cdots & \alpha^{(\nmind)}_\nvar   
\emx \in \bbZp^{\nvar \times \nmind}.
\end{equation}
For a set $\mathfrak{A}\subset \bbZp^{\nvar}$ we write $\calA \sim \mathfrak{A}$, if $\mathfrak{A}$ is the set of columns of $\calA$. For each set of multidegrees $\mathfrak{A}$ there exist $m!$ matrix representations (i.e. $\calA$ such that $\calA \sim \mathfrak{A}$). Each representation defines an ordering of multidegrees.

The matrix $\calA$ defines the vector of monomials
\begin{equation}\label{eq:monomials_calA}
\phi_{\calA}(d) :=  \bmx \phi_1(d) & \cdots & \phi_{\nmind} (d) \emx^{\top},\quad \phi_k(\dpoint) := \dpoint^{\alpha^{(k)}}.
\end{equation}
\begin{example}\label{ex:conic2}
The following integer matrix
\begin{equation}\label{eq:calAconic}
\calA = 
\bmx
2 & 1 & 0 & 1 & 0 & 0 \\
0 & 1 & 2 & 0 & 1 & 0 \\
\emx
\end{equation}
corresponds to the vector of monomials \eqref{eq:conic_monomials}.
In this case, algebraic hypersurface fitting coincides with conic section fitting. (See also Fig.~\ref{fig:conic}.)
Also, note that $\calA \sim \trgset{2}{2}$. This set of multidegrees is depicted in Fig.~\ref{fig:calAconic}.
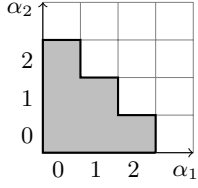
\begin{figure}[ht!]
\centering
\begin{tikzpicture}
\begin{scope}
    \draw (0,0) node[anchor=north west]  {\small$0$};
    \draw (0.5,0) node[anchor=north west]  {\small$1$};
    \draw (1,0) node[anchor=north west]  {\small$2$};
    \draw (0,0) node[anchor=south east]    {\small$0$};
    \draw (0,0.5) node[anchor=south east]    {\small$1$};
    \draw (0,1) node[anchor=south east]    {\small$2$};
    \draw (0,1.7) node[anchor=south east]    {\small$\alpha_2$};
    \draw (1.6,-0.05) node[anchor=north west]    {\small$\alpha_1$};
        
    \draw[color=black, help lines, line width=.1pt] (0,0) grid[xstep=0.5cm, ystep=0.5cm] (2,2);
    \draw[fill=lightgray, thick] (1.5,0) -- (1.5,0.5) -- (1,0.5) -- (1,1) -- (0.5,1) -- (0.5, 1.5) -- (0,1.5) -- (0, 0) -- (1.5,0);
    \draw [->] (0,0) -- (0,2);
    \draw [->] (0,0) -- (2,0);
\end{scope}
\end{tikzpicture}
\caption{Set of multidegrees  $\trgset{2}{2}$.}%
\label{fig:calAconic}
\end{figure}
\end{example}

\subsection{Ordinary least squares estimator}
Now consider the OLS estimator defined in \eqref{eq:ols_opt_problem}. First, we rewrite the cost function in a matrix form. For a set of points $\dpoints = \{\dpoint^{(1)}, \ldots, \dpoint^{(\npoints)}\} \subset \bbR^{\nvar}$, we define the
\textit{multivariate Vandermonde matrix} \cite{Mourrain.Pan00Joc-Multivariate}
as
\begin{equation}\label{eq:vandmat_def}
\vandmat{}{\dpoints} \defeq
\bmx \phi(\dpoint^{(1)}) & \cdots & \phi(\dpoint^{(\npoints)})\emx \in \bbR^{\nmind \times \npoints}.
\end{equation}
Then we have that the vector of the residuals can be expressed as
\[
\bmx R_{\theta}(\dpoint^{(1)}) & \cdots & R_{\theta}(\dpoint^{(\npoints)}) \emx =  \theta^{\top} \vandmat{}{\dpoints}.
\]
Therefore, the OLS cost function \eqref{eq:ols_cost_def} is equal to 
\begin{equation}\label{eq:VtV}
Q_{\text{ols}}(\theta, \dpoints) = \theta^{\top} \Psi(\dpoints) \theta, \quad \mbox{where } \Psi(\dpoints) = \vandmat{}{\dpoints} \left(\vandmat{}{\dpoints}\right)^{\top}.
\end{equation}
Now we consider the case of weighted $2$-norm defined as 
\begin{equation}\label{eq:wnorm}
\|\theta\|^2_w = \sum\limits_{j=1}^{\nmind} w_j \theta^2_j,
\end{equation}
where $w_j \in (0,\infty)$. Then ${\widehat{\theta}}_{\text{ols}}$ is given as a solution of an eigenvalue problem.
\begin{lemma}\label{lem:ols_solution}
For the weighted $2$-norm \eqref{eq:wnorm}, the solution of \eqref{eq:ols_opt_problem} is given by
\[
{\widehat{\theta}}_{\text{ols}} = \Lambda {\widehat{\theta}'}_{\text{ols}},
\]
where ${\widehat{\theta'}}_{\text{ols}}$ is an eigenvector of the symmetric matrix $\Lambda \Psi(\dpoints) \Lambda$ corresponding to its smallest eigenvalue, and $\Lambda := \diag\left(w^{-\frac{1}{2}}_1, \ldots, w^{-\frac{1}{2}}_{\nmind}\right)$. Equivalently ${\widehat{\theta}'}_{\text{ols}}$ is a left singular vector of the matrix $\Lambda \vandmat{}{\dpoints}$ corresponding to its smallest singular value.
\end{lemma}
\begin{proof}
After a change of variables $\theta = \Lambda \theta'$, the problem \eqref{eq:ols_opt_problem} becomes
\begin{equation}\label{eq:ols_opt_problem_scaled}
\begin{split}
{\widehat{\theta}'}_{\text{ols}} := & \argmin_{\theta' \in \bbR^{m}} \theta'^{\top} \Lambda  \Psi(\dpoints)\Lambda  \theta'  \, \\
                          & \sto \|\theta' \|^2_2 = 1.
\end{split}
\end{equation}
The solution of \eqref{eq:ols_opt_problem_scaled} is given by an eigenvector of the symmetric matrix
$\Lambda \Psi(\dpoints) \Lambda$ corresponding to its smallest eigenvalue.
\end{proof}
Note that the OLS estimator has the following properties.
\begin{note}\label{not:exact_recovery}
In the model \eqref{eq:eiv}, if $\widetilde{\dpoint}^{(j)} \equiv 0$, then $\Psi(\dpoints) = \Psi(\overline{\dpoints})$, where $\overline{\dpoints} = \{\overline{\dpoint}^{(1)}, \ldots, \overline{\dpoint}^{(\npoints)}\}$. Therefore, in this case, if the solution to \eqref{eq:ols_opt_problem} is unique, 
then $\widehat{\theta}_{\text{ols}}$ is equal to $\overline{\theta}$ (up to scaling). In other words, in a non-degenerate noiseless case the OLS estimator recovers the true parameter vector $\overline{\theta}$.
\end{note}
The uniqueness of the solution in Note~\ref{not:exact_recovery} corresponds to uniqueness of the algebraic hypersurface that contains the true data points.
For example, for a set of $4$ points in a general position there is a nonunique conic section
passing through them (see an example in \cite{kpca}).

\begin{note}\label{not:bias}
If the data is noisy, and $\widetilde{\dpoint}^{(j)} \sim \normal(0, \Sigma)$, then, in general,
\begin{equation}\label{eq:bias}
\bfE (\Psi(\dpoints)) \neq \Psi(\overline{\dpoints}),
\end{equation}
where $\bfE(\cdot)$ denotes the mathematical expectation.
\end{note}
Note~\ref{not:bias} gives an explanation why the OLS estimator is inconsistent. 

\subsection{Deconvolution of polynomials}
The construction of the adjusted least squares estimator is based on finding a matrix $\psimatcor{\calA}{\dpoints}{\sigma}$, such that its expectation is equal to $\Psi(\overline{\dpoints})$ in the noise model \eqref{eq:eiv} and \eqref{eq:gaussian_errors} (compare with \eqref{eq:bias}).
For this purpose, following \cite{Shklyar.etal07JoMA-conic}  and \cite{Shklyar09PhD-Consistency}, we introduce the operation of deconvolution.
\begin{definition}\label{def:deconvolution}
For a multivariate polynomial $f$ and a positive-semidefinite covariance matrix $\Sigma$, the deconvolution is defined as
\[
f\ast p_{-\Sigma} \defeq g, \quad \mbox{ where } \mathbf{E}(g(a+x)) = f(a),\quad \mbox{for all } a \in \bbR^{q} \mbox{ and } x \sim \normal(0, \Sigma).
\]
\end{definition}
The deconvolution operation has the following properties \cite[\S 5.1]{Shklyar.etal07JoMA-conic}.
\begin{lemma}\label{lem:deconvolution_prop}
\begin{enumerate}
\item For any polynomial $f$, its deconvolution $f \ast p_{-\Sigma}$ is a polynomial.
\item Deconvolution is linear, i.e., 
\[
(f_1+ cf_2) \ast p_{-\Sigma} = f_1 \ast p_{-\Sigma} + c (f_2 \ast p_{-\Sigma}).
\]
\item For an affine transformation $T(d) = Kd+b$, $K\in \bbR^{\nvar \times \nvar}$, $b \in \bbR^{\nvar}$, we have
\[
(f \circ T) \ast p_{-\Sigma} = (f \ast p_{-K\Sigma K^{\top}}) \circ T,
\]
where $f\circ T$ is a composition of $f$ and $T$, i.e., $(f\circ T)(d) = f(T(d))$
\end{enumerate}
\end{lemma}
In Section~\ref{sec:als_usigma_constructive} we give an explicit form of deconvolution of monomials with respect to $\Sigma = \sigma^2 I$.

\subsection{Adjusted matrix $\Psi$ and ALS estimator for known variance}
For a covariance matrix $\Sigma = \sigma^2 \Sigma_0$, the \emph{adjusted matrix} $\psimatcor{}{\dpoints}{\sigma} \in \bbR^{\nmind \times \nmind}$ is defined as \cite{Shklyar09PhD-Consistency}
\[
(\psimatcor{}{\dpoints}{\sigma})_{i,j} = \sum\limits_{k=1}^{N} (\phi_i\phi_j) \ast p_{-\sigma^2 \Sigma_0} (\dpoint^{(k)}),
\]
where $\phi_i\phi_j$ is the product of polynomials $\phi_i$ and $\phi_j$.
By Definition~\ref{def:deconvolution}, we have that for ${\dpoints}$ generated according to \eqref{eq:eiv} and \eqref{eq:gaussian_errors}, the equation
\begin{equation}\label{eq:psi_adj_me}
\bfE (\psimatcor{}{\dpoints}{\sigma}) = \Psi(\overline{\dpoints})
\end{equation}
holds true for any set of true points $\overline{\dpoints}$. 

Then the first version of the ALS estimator (for the case of known $\sigma$) is defined as 
\begin{equation}\label{eq:als1_opt_problem}
\begin{split}
{\widehat{\theta}}_{\text{als},\sigma} := & \argmin_{\theta \in \bbR^{m}} Q_{\text{als},\sigma}(\theta,\dpoints),  \, \\
                          & \sto \|\theta \| = 1,
\end{split}
\end{equation}
where the cost function $Q_{\text{als},\sigma}$ is 
\begin{equation}\label{eq:als_cost_def}
Q_{\text{als},\sigma}(\theta,\dpoints) := \theta^{\top}\psimatcor{}{\dpoints}{\sigma} \theta.
\end{equation}
In \cite{Shklyar.etal07JoMA-conic}, this version of ALS estimator is denoted by $\widehat{\theta}_{\text{als}1}$.

Since $Q_{\text{als},\sigma}$ is a quadratic form, we have that the following lemma can be proved analogously to Lemma~\ref{lem:ols_solution}.
\begin{lemma}
In the case of the weighted $2$-norm \eqref{eq:wnorm}, the ALS estimator for known variance is given by  ${\widehat{\theta}}_{\text{als},\sigma} = \Lambda {\widehat{\theta}'}_{\text{als},\sigma}$, where is ${\widehat{\theta}'}_{\text{als},\sigma}$ is an eigenvector of the symmetric matrix $\Lambda \psimatcor{\calA}{\dpoints}{\sigma}  \Lambda$ corresponding to its smallest eigenvalue. 
\end{lemma}
Note that unlike the matrix $\Psi(\dpoints)$, the matrix $\psimatcor{\calA}{\dpoints}{\sigma}$  cannot, in general, be factorized as $B(\dpoints) B^{\top}(\dpoints)$. Moreover, $\psimatcor{\calA}{\dpoints}{\sigma}$ may be indefinite or negative semidefinite, thus the smallest eigenvalue of $\psimatcor{\calA}{\dpoints}{\sigma}$ may be negative.

\subsection{ALS estimator for unknown variance: an abstract definition}
A more important case is when the variance is not known, i.e., when $\Sigma = \sigma^2 \Sigma_0$, and we know only $\Sigma_0$. In \cite{Kukush.etal04CSDA-Consistent}, it was proposed to estimate $\theta$ and $\sigma$ simultaneously (for quadratic hypersurfaces). In  \cite{Shklyar09PhD-Consistency}, this definition was extended to the general class of algebraic hypersurfaces (defined by \eqref{eq:alg_hyp}).
 
The second version of the ALS estimator (with unknown $\sigma^2$) is constructed as follows:
$\widehat{\sigma}$ is a solution of \eqref{eq:als2_equation},
and $\widehat{\theta}_{\text{als}}$ is defined as a solution of
\[
\psimatcor{\calA}{\dpoints}{\widehat{\sigma}} \widehat{\theta}_{\text{als}} = 0, \quad\|\widehat{\theta}_{\text{als}}\| = 1.
\]
In \cite{Shklyar09PhD-Consistency}, many important properties of $\widehat{\theta}_{\text{als}}$ are proved under the following assumption.
\begin{assumption}\label{ass:closedness}
The set of polynomials in \eqref{eq:monomial_vec} is closed under the operation of taking partial derivatives,
i.e., for each $i = 1,\ldots,\nvar$ there exists a matrix $D_i \in \bbR^{\nmind \times \nmind}$ such that
\[
\frac{\partial}{\partial d_i} \phi(d) = D_i \phi(d).
\]
\end{assumption}
Note that, if $\phi$ is given by the matrix of multidegrees \eqref{eq:monomials_calA}, with $\calA \sim \mathfrak{A}$, then Assumption~\ref{ass:closedness} holds if and only if $\mathfrak{A}$ is a lower set \eqref{eq:lower_set_def}. Indeed, if $\alpha \in \mathfrak{A}$, then
\[
\frac{\partial}{\partial d_i} d^{\alpha} = 
\begin{cases}
0, &  \alpha_i = 0, \\
\alpha_i d^{\alpha - e_i}, &  \alpha_i > 0. \\
\end{cases}
\]
In the latter case, if $\mathfrak{A}$ is a lower set, then $\alpha - e_i \in \mathfrak{A}$ for any $i$.
For example, in Example~\ref{ex:conic2}, 
\[
\frac{\partial}{\partial d_1} \phi_{\calA} (d) = 
\bsm
0 & 0 & 0 & 2 & 0 & 0 \\
0 & 0 & 0 & 0 & 1 & 0 \\
0 & 0 & 0 & 0 & 0 & 0 \\
0 & 0 & 0 & 0 & 0 & 1 \\
0 & 0 & 0 & 0 & 0 & 0 \\
0 & 0 & 0 & 0 & 0 & 0 
\esm \phi_{\calA} (d), \quad
\frac{\partial}{\partial d_2} \phi_{\calA} (d) = 
\bsm
0 & 0 & 0 & 0 & 0 & 0 \\
0 & 0 & 0 & 1 & 0 & 0 \\
0 & 0 & 0 & 0 & 2 & 0 \\
0 & 0 & 0 & 0 & 0 & 0 \\
0 & 0 & 0 & 0 & 0 & 1 \\
0 & 0 & 0 & 0 & 0 & 0 
\esm \phi_{\calA} (d).
\]

The next result shows that under Assumption~\ref{ass:closedness} and mild additional conditions, the solution of~\eqref{eq:als2_equation} exists and is unique.
\begin{theorem}[See {\cite[Theorem 3.4]{Shklyar09PhD-Consistency}}]
Assume that the vector $\phi(d)$ satisfies Assumption~\ref{ass:closedness}. Then the following statements hold true.
\begin{itemize}
\item If $\Sigma_0$ is positive definite, then the equation~\eqref{eq:als2_equation} has a unique solution.
\item If $\Sigma_0$ is rank-deficient, then
\begin{itemize}
\item if the matrix $\Psi(\dpoints)$ is positive definite, then the equation~\eqref{eq:als2_equation} has at most one solution;
\item if there exist vectors $h \in \bbR^{\nmind}$, $a \in  \bbR^{\nvar}$ and a scalar $b \in \bbR$, such that
\[
h^{\top} \phi(d) = a^{\top} d + b, \quad \Sigma_0 h \neq 0
\]
then the equation \eqref{eq:als2_equation} has at least one solution.
\end{itemize}
\end{itemize}
\end{theorem}
\begin{corollary}[See {\cite[Corollary 3.5]{Shklyar09PhD-Consistency}}]\label{cor:existence}
If the solution of \eqref{eq:als2_equation} exists, is unique and is equal to $\widehat{\sigma}$,
then 
\begin{itemize}
\item $\lambda_{min}(\psimatcor{\calA}{\dpoints}{\sigma}) > 0$ for $0 \le \sigma < \widehat{\sigma}$, and
\item $\lambda_{min}(\psimatcor{\calA}{\dpoints}{\sigma}) < 0$ for $\sigma > \widehat{\sigma}$.
\end{itemize}
\end{corollary}

Next, in \cite{Shklyar09PhD-Consistency} it was also proved that under Assumption~\ref{ass:closedness} and some conditions on the true data, the estimator is strongly consistent.
\begin{theorem}[See {\cite[Theorem 3.14]{Shklyar09PhD-Consistency}}]\label{thm:consistency}
Let 
\[
\dpoint^{(1)}, \dpoint^{(2)}, \ldots, \dpoint^{(N)}, \ldots,
\]
be an infinite sequence of points generated as in \eqref{eq:eiv}, and 
\[\widehat{\sigma}_N := \widehat{\sigma}(\dpoints_N) \quad\text{and}\quad \widehat{\theta}_N:= \widehat{\theta}_{\text{als}}(\dpoints_N)\] 
denote the ALS estimators for the first $\npoints$ data points $\dpoints_{\npoints} = \{\dpoint^{(1)}, \dpoint^{(2)}, \ldots, \dpoint^{(N)}\}$.

Also assume that the true points $\overline{\dpoint}^{(n)}$ satisfy the following conditions
\begin{equation}\label{eq:cc_data_finitness}
\begin{split}
\mbox{for any } j, \quad& \sum\limits_{n=1}^{\infty} \frac{1}{n^2}
\left\|\left.\frac{\partial}{\partial d_j} \Psi\left(\{d\}\right)\right|_{d=\overline{\dpoint}^{(n)}}
\right\|^2_F < \infty, \\
\mbox{for any } j_1, j_2, \quad& \sup_{n \ge1} \frac{1}{n} \left\|
\left.\frac{\partial^2}{\partial d_{j_1} d_{j_2}}  \Psi\left(\{d\}\right)\right|_{d=\overline{\dpoint}^{(n)}} \right\|^2_F < \infty,
\end{split}
\end{equation} 
and
\begin{equation}\label{eq:cc_distribution}
\lim_{N\to \infty}  \frac{1}{N} \lambda_{2} \left( \Psi(\overline{\dpoints}_{\npoints}) \right) > 0,
\end{equation}
where $\lambda_2$ is the second smallest eigenvalue of a matrix.

Let $\overline{\theta}$ be the true parameter vector, such that 
\begin{equation}\label{eq:cc_noise_identifiability}
\mbox{the polynomial } f(d,u):=\overline{\theta}^{\top} \phi(d + \Sigma_0 u),
\mbox{ depends on the variable } u.
\end{equation}
If the conditions \eqref{eq:cc_data_finitness}, \eqref{eq:cc_distribution} and \eqref{eq:cc_noise_identifiability} are satisfied,
\[
\widehat{\sigma}_N \to \sigma, \quad \sin \angle(\widehat{\theta}_N, \overline{\theta}) \to 0,
\]
where the convergence is almost surely.
\end{theorem}
\begin{note}
The conditions of the Theorem~\ref{thm:consistency} are rather mild.
\begin{itemize}
\item
The condition \eqref{eq:cc_data_finitness} is on boundedness of the data. For example, it is satisfied if all the true data points are within a bounded region. 
\item
The condition \eqref{eq:cc_distribution} ensures that the data points are well-distributed. For example, if the true hypersurface is a union of two hyperplanes, the condition \eqref{eq:cc_distribution} ensures that there are sufficiently many true points on both hyperplanes.
\item
 The condition \eqref{eq:cc_noise_identifiability} means that the noisy vectors do not lie inside the true hypersurface. Indeed, if \eqref{eq:cc_noise_identifiability}  is not satisfied, then $\Psi(\dpoints) = \Psi(\overline{\dpoints})$, and $\overline{\theta}^{\top} \psimatcor{\calA}{\dpoints}{\sigma} = 0$. In this case, $\overline{\theta}$ still can be recovered, but the noise variance $\sigma^2$ cannot.
\end{itemize}
\end{note}

\subsection{ALS estimator for unknown variance: a constructive approach}\label{sec:als_usigma_constructive}
Now we recall the algorithm of \cite[Ch. 6]{Markovsky12-Low}, for computing the ALS estimators in the case $\Sigma_0 = I$ and $\phi = \phi_{\calA}$ is given as in \eqref{eq:monomials_calA}\footnote{In \cite[Ch. 6]{Markovsky12-Low} it was assumed that $\calA \sim \trgset{\nvar}{r}$, but this assumption is not necessary.}.
The construction of the ALS estimators is based on \emph{homogeneous Hermite polynomials}, defined as 
\begin{align*}
\herm{\sigma}{0}{z} & =  1, \\ 
\herm{\sigma}{1}{z} & =  z, \\ 
\herm{\sigma}{k}{z} & =  z \herm{\sigma}{k-1}{z} - (k-1) \sigma^2 \herm{\sigma}{k-2}{z}.  
\end{align*}
The key property of the homogeneous Hermite polynomials is the following {deconvolution} property. 
\begin{lemma}[See {\cite[Ch. 6]{Markovsky12-Low}}]\label{lem:deconvolution}
If $\varepsilon \sim \normal(0,\sigma^2)$, then
\begin{equation}\label{eq:conv}
\bfE (\herm{\sigma}{k}{a+\varepsilon}) = a^k
\end{equation}
for any $a \in \bbR$.
\end{lemma}
\begin{corollary}\label{cor:deconvolution_monomials}
If we define
\begin{equation}\label{eq:monomial}
f_{\beta}(d) = d_1^{\beta_1} \cdots d_{\nvar}^{\beta_\nvar}
\end{equation}
then the deconvolution of a monomial is 
\begin{equation}\label{eq:deconvolution_monomials}
f_{\beta} \ast p_{-\sigma^2I} =  \herm{\sigma}{\beta_1}{\dpoint_1}\herm{\sigma}{\beta_2}{\dpoint_2} \cdots \herm{\sigma}{\beta_{\nvar}}{\dpoint_\nvar}.
\end{equation}
\end{corollary}
Using Corollary~\ref{cor:deconvolution_monomials}, we can construct the adjusted matrix $\Psi$  by replacing all monomials in $\Psi(\dpoints)$ by the corresponding polynomials from \eqref{eq:deconvolution_monomials}. More precisely,
from \eqref{eq:VtV}, the $(k,l)$-th element of $\Psi(\dpoints)$ is
\begin{equation}\label{eq:psimat_elem}
(\Psi(\dpoints))_{k,l} = \sum\limits_{j=1}^{\npoints} \phi_{k}(\dpoint^{(j)})\phi_{l}(\dpoint^{(j)}) =
\sum\limits_{j=1}^{\npoints} {(\dpoint^{(j)}_1)}^{\beta_1} \cdots {(\dpoint^{(j)}_{\nvar})}^{\beta_{\nvar}},
\end{equation}
where $\beta = \alpha^{(k)} + \alpha^{(l)}$. Then the  $(k,l)$-th element of $\psimatcor{\calA}{\dpoints}{\sigma}$ is equal to 
\begin{equation}\label{eq:psimatsig_def}
(\psimatcor{\calA}{\dpoints}{\sigma})_{k,l}  =
\sum\limits_{j=1}^{\npoints} \herm{\sigma}{\beta_1}{\dpoint^{(j)}_1}\herm{\sigma}{\beta_2}{\dpoint^{(j)}_2} \cdots \herm{\sigma}{\beta_{\nvar}}{\dpoint^{(j)}_\nvar},
\end{equation}
From \eqref{eq:psimatsig_def}, the matrix $\psimatcor{\calA}{\dpoints}{\sigma}$ has the form
\begin{equation}\label{eq:mpoly}
\psimatcor{\calA}{\dpoints}{\sigma} = \Psi(\dpoints) + \sigma^2 \Psi_1(\dpoints) + \cdots + 
\sigma^{2r} \Psi_r(\dpoints),
\end{equation}
where $r$ is the degree of the polynomial $R_{\theta}(d)$ (i.e., the maximal total degree of $\phi_j(d)$) and $\Psi_k(\calD)$ do not depend on $\sigma$. Indeed, only even powers of $\sigma$ are present in $\psimatcor{\calA}{\dpoints}{\sigma}$, and the highest power corresponds  to the highest total degree of a monomial in $\Psi(\dpoints)$, which is equal to  $2r$. Note also that by Corollary~\ref{cor:existence} it follows that $\widehat{\sigma}$ is the smallest $\sigma$ such that $\psimatcor{\calA}{\dpoints}{\sigma}$ is rank-deficient. Thus $\widehat{\sigma}^2$ is equal to the smallest polynomial eigenvalue of the matrix polynomial \eqref{eq:psimatsig_def}, and $\widehat{\theta}_{\text{als}}$ is its corresponding eigenvector. Thus the solution of the polynomial eigenvalue problem given in \cite[Ch. 6]{Markovsky12-Low} computes the estimator defined in \cite{Shklyar09PhD-Consistency}.

%% file: theory.tex
\section{Computation of the ALS estimators and existence of solutions}\label{sec:comp_adjusted}
In this section, we construct the matrix polynomial \eqref{eq:mpoly} for an arbitrary set of basis polynomials $\phi(d)$ and arbitrary $\Sigma_0$. For the case when $\phi(d)$ is a vector of monomials, we show that the matrices $\Psi_j(\dpoints)$ are quasi-Hankel and 
can be constructed using simple operations on the moment array of data.

\subsection{Reduction to the simple case}
In this subsection, we show how the general case can be reduced to the case similar to the one discussed in Section~\ref{sec:als_usigma_constructive}. First, let $\Sigma_0$ be of rank $s$. Then there exists a nonsingular matrix $K \in \bbR^{\nvar \times \nvar}$ such that 
\[
\Sigma_0 = K J_s K^{\top},
\] 
where $J_{\nvar} = I_{\nvar}$ and
\begin{equation}\label{eq:normalized_sigma}
J_s := \bmx I_s & 0 \\ 0 & 0 \emx \in \bbR^{\nvar \times \nvar}\quad\text{for}\quad s < \nvar.
\end{equation}
Now consider the linear transformation of data $T(d) = K^{-1} d$. We have that
\[
\theta^{\top} \phi(d) = \theta^{\top} \phi^{(K)} (T(d)),
\]
where $\phi^{(K)}(d)$ is the transformed vector of basis polynomials
\[
\phi^{(K)}(d):= \phi\circ T^{-1}(d) = \phi(Kd).
\]
Next, if $\widetilde{d} \sim \normal(0, \sigma^2\Sigma_0)$, then $T(\widetilde{d}) \sim \normal(0, \sigma^2 J_s)$. Finally, by Lemma~\ref{lem:deconvolution_prop}, we have that $(f \circ T^{-1}) \ast p_{-\sigma^2 J_s}(T(d)) = f\ast p_{-\sigma^2 \Sigma_0} (d)$, and therefore
\[
\psimatcor{}{\dpoints}{\sigma} =  \psimatcor{}{T(\dpoints)}{\sigma}',
\] 
where $\psimatcor{}{T(\dpoints)}{\sigma}'$ denotes the adjusted matrix for the transformed covariance matrix $\sigma^2 J_s$ and transformed basis polynomials $\phi^{(K)}$. We can summarize these observations as follows.
\begin{note}
Without loss of generality, we can assume that $\Sigma_0 = J_s$. For general $\Sigma_0$, we can always transform the problem to the case $\Sigma_0 = J_s$ by a nonsingular linear transformation of data.
\end{note}
Now assume that $\Sigma_0 = J_s$. For any $\phi$, then there exists a multidegree matrix $\calA \in \bbZp^{\nvar\times \nmind_2}$ and the matrix $F \in \bbR^{\nmind \times \nmind_2}$ such that $\phi(d) = F\phi_{\calA} (d)$, where $\phi_{\calA}$ is defined in \eqref{eq:monomials_calA}. Then we have that
\[
\Psi(\dpoints) = F \Psi_{\calA} (\dpoints) F^{\top},
\]
where $\Psi_{\calA}$ is the matrix $\Psi$ for the vector of basis polynomials $\phi_{\calA}$ given in \eqref{eq:monomials_calA}.
By linearity of the deconvolution operation, we have that 
\[
\psimatcor{}{\dpoints}{\sigma} = F \psimatcor{}{\dpoints}{\calA,\sigma}F^{\top},
\]
where $\psimatcor{}{\dpoints}{\calA,\sigma}$ is the adjusted matrix for the vector of monomials $\phi_{\calA}$.

Now, assume that $\Sigma_0 = J_s$  and $\phi = \phi_{\calA}$ is given as a vector of monomials \eqref{eq:monomials_calA}.
We have that an analogue of Corollary~\ref{cor:deconvolution_monomials} holds.
 \begin{corollary}\label{cor:deconvolution_monomials2}
For a monomial $f_{\beta}(d)$ defined in \eqref{eq:monomial} and $\Sigma_0 = J_s$, the deconvolution of a monomial is equal to
\begin{equation}\label{eq:deconvolution_monomials2}
f \ast p_{-\sigma^2J_s} =  \herm{\sigma}{\beta_1}{\dpoint_1}\cdots \herm{\sigma}{\beta_{s}}{\dpoint_s} \dpoint_{s+1}^{\beta_{s+1}} \cdots \dpoint_{\nvar}^{\beta_\nvar}.
\end{equation}
\end{corollary}
From Corollary~\ref{cor:deconvolution_monomials2}, we can compute the adjusted matrix $\Psi$ as in Section~\ref{sec:als_usigma_constructive}. Indeed, the  $(k,l)$-th element of $\psimatcor{\calA}{\dpoints}{\sigma}$ is equal to 
\begin{equation}\label{eq:psimatsig_def2}
(\psimatcor{\calA}{\dpoints}{\sigma})_{k,l}  = 
\sum\limits_{j=1}^{\npoints} \herm{\sigma}{\beta_1}{\dpoint^{(j)}_1} \cdots \herm{\sigma}{\beta_{s}}{\dpoint^{(j)}_s} {\left(\dpoint^{(j)}_{s+1}\right)}^{\beta_{s+1}} \cdots {\left(\dpoint^{(j)}_{\nvar}\right)}^{\beta_\nvar},
\end{equation}
where $\beta = \alpha^{(k)}+ \alpha^{(l)}$. Therefore, the case $\Sigma_0 = J_s$ is analogous to the case considered in Section~\ref{sec:als_usigma_constructive}.
In particular, we have that $\psimatcor{\calA}{\dpoints}{\sigma}$ has the form \eqref{eq:mpoly}, 
where $r = \max |\alpha^{(k)}|_s$ and $|\cdot|_s$ is defined in \eqref{eq:total_degree}.

In the rest of this section, we assume that $\Sigma_0 = J_s$ and $\phi = \phi_{\calA}$ is the vector of monomials defined in \eqref{eq:monomials_calA}.

\subsection{Quasi-Hankel matrices}\label{sec:hankel}
Now we recall the definition of a class of structured matrices that is one of the key ingredients of this paper. Let $\rmB  = \bmx\rmB_{\alpha}\emx_{\alpha \in \bbZp^{\nvar}}$  be an infinite $\nvar$-way array and $\calA$ be a $\nvar\times \nmind$ integer matrix, as in \eqref{eq:calA}. Then the symmetric \textit{quasi-Hankel} \cite{Mourrain.Pan00Joc-Multivariate} matrix $\scrH_{\calA} (\rmB)$, constructed from $\calA$ and $\rmB$ is the following $\nmind \times \nmind$ matrix: 
\[
\big(\scrH_{\calA} (\rmB)\big)_{k,l} = \rmB_{\alpha^{(k)}+\alpha^{(l)}}.
\]
The rows and columns in the symmetric quasi-Hankel matrix correspond to multidegrees from $\calA$.
\begin{note}\label{not:qhank_support}
Let $\mathfrak{A}$ be the set of columns of the matrix $\calA$ (i.e., $\calA \sim \mathfrak{A}$).
Then for construction of $\scrH_{\calA} (\rmB)$ only the elements $\rmB_{\alpha}$ with $\alpha \in \mathfrak{A} + \mathfrak{A}$ are needed.
\end{note}

\begin{example}
Consider a $1$-dimensional ($\nvar = 1$) array $\rmB = \bmx \rmB_{0} & \rmB_{1} & \cdots  \emx^{\top}$, and fix the sets 
\[
\calA = \bmx0 &\cdots&k\emx \sim \boxset{k}.
\]
Then the quasi-Hankel matrix is $\scrH_{\calA} (\rmB) = \calH_{k} (\rmB)$, where
\[
\quad \calH_{k} (\rmB) :=
\bmx
 \rmB_{0}   &  \rmB_{1}   &  \cdots &  \rmB_{k}   \\
 \rmB_{1}   &  \rmB_{2}   &  \cdots &  \rmB_{k+1}   \\
\vdots      &  \vdots     &         &  \vdots    \\
 \rmB_{k}   &  \rmB_{k+1} &  \cdots &  \rmB_{2k} 
\emx,
\]
is the ordinary square Hankel matrix for the sequence $\rmB$. In $\calH_{k} (\rmB)$, only the elements $\rmB_j$, $0 \le j \le 2k$, are used.
\end{example}

\begin{example}
Consider a $2$-dimensional ($\nvar = 2$) array $\rmB = \bmx \rmB_{[i,j]} \emx_{[i,j]\in \bbZp^2}$, and fix the set
\[
\calA = 
\bmx 
0 & \cdots & k & 0 & \cdots & k & \cdots & 0 & \cdots & k \\ 
0 & \cdots & 0 & 1 & \cdots & 1 & \cdots & l & \cdots & l 
\emx \sim  \boxset{[k \ l]^{\top}}.
\]
Then the quasi-Hankel matrix is a symmetric \emph{Hankel-block-Hankel} matrix:
\[
\scrH_{\calA} (\rmB) = 
\bmx
 \calH_{k}(\rmB_{[:,0]}) & \calH_{k}(\rmB_{[:,1]})   & \cdots &  \calH_{k}(\rmB_{[:,l]})   \\
 \calH_{k}(\rmB_{[:,1]}) & \calH_{k}(\rmB_{[:,2]})   & \cdots &  \calH_{k}(\rmB_{[:,l+1]}) \\
\vdots                   & \vdots                    &        &  \vdots    \\
 \calH_{k}(\rmB_{[:,l]}) & \calH_{k}(\rmB_{[:,l+1]}) & \cdots &  \calH_{k}(\rmB_{[:,2l]}) 
\emx,
\]
i.e. a block-Hankel matrix with Hankel blocks constructed from the columns $\rmB{[:,j]}$ of  $\rmB$.
In the case $\nvar > 2$  and $\calA \sim  \boxset{\gamma}$ with $\gamma\in \bbZp^{\nvar}$  (given in the vectorization order), the matrix $\scrH_{\calA} (\rmB)$ is a multilevel Hankel matrix \cite{Fasino.Tilli00LAaiA-Spectral}.
\end{example}

It is easy to see that the matrices $\Psi(\dpoints)$ and $\psimatcor{\calA}{\dpoints}{\sigma}$ are quasi-Hankel.
\begin{lemma}
\begin{enumerate}
\item
The matrix $\Psi(\dpoints)$ is a symmetric quasi-Hankel matrix
\[
\Psi(\dpoints) = \scrH_{\calA} (\mCoef)
\]
where $\mCoef = \bmx \mcoef_\alpha\emx_{\alpha \in \bbZp^q}$ is the infinite \textit{moment array} defined as
\[
\mcoef_\alpha \defeq \sum\limits_{j=1}^{\npoints} (\dpoint^{(j)})^{\alpha}.
\]

\item
For $\Sigma_0 = J_s$, the matrix  $\psimatcor{\calA}{\dpoints}{\sigma}$ is quasi-Hankel
\begin{equation}\label{eq:psimatcor_qhankel}
\psimatcor{\calA}{\dpoints}{\sigma} = \scrH_{\calA}(\mCoefSig), 
\end{equation}
where $\mCoefSig = \bmx \mcoefSig_{\alpha} \emx_{\alpha \in \bbZp^{\nvar}}$ is the \textit{$\sigma$-adjusted moment array}, defined as
\begin{equation}\label{eq:mcoefsig_def}
\mcoefSig_\alpha \defeq 
\sum\limits_{j=1}^{\npoints} \herm{\sigma}{\alpha_1}{\dpoint^{(j)}_1} \cdots \herm{\sigma}{\alpha_{s}}{\dpoint^{(j)}_s}
{\left(\dpoint^{(j)}_{s+1}\right)}^{\alpha_{s+1}} \cdots {\left(\dpoint^{(j)}_{\nvar}\right)}^{\alpha_q},
\end{equation}
such that $\dpoint^{(j)} = \bmx \dpoint^{(j)}_1 & \cdots& \dpoint^{(j)}_\nvar \emx^{\top}$.
\end{enumerate}
\end{lemma}
\begin{proof}
Follows immediately from \eqref{eq:psimat_elem} and \eqref{eq:psimatsig_def}.
\end{proof}

\subsection{Coefficients of Hermite polynomials and array shifts}\label{sec:comp_M}
For convenience, we denote the coefficients of the Hermite polynomials as
\[
\herm{\sigma}{k}{z} = \sum\limits_{i+j=k} \rmH_{[i,j]} \sigma^{i} z^j.
\]
Then the coefficients of all Hermite polynomials can be arranged in the infinite array $\rmH = \bmx\rmH_{[i,j]}\emx_{[i, j] \in \bbZp^2} $. In Table~\ref{tab:herm}, a part of the infinite array $\rmH$ is shown. 

\begin{table}[!hbt]
\caption{Table of coefficients $\rmH_{[i,j]}$ of the infinite array $\rmH$. Row --- $i$, column --- $j$.}
\label{tab:herm}
\begin{center}
\begin{filecontents*}{runGenTable.out.txt}
1 1 1 1 1 1 1 1 1
0 0 0 0 0 0 0 0 NaN
-1 -3 -6 -10 -15 -21 -28 NaN NaN
0 0 0 0 0 0 NaN NaN NaN
3 15 45 105 210 NaN NaN NaN NaN
0 0 0 0 NaN NaN NaN NaN NaN
-15 -105 -420 NaN NaN NaN NaN NaN NaN
0 0 NaN NaN NaN NaN NaN NaN NaN
105 NaN NaN NaN NaN NaN NaN NaN NaN
\end{filecontents*}
\pgfplotstabletypeset[every head row/.style={before row={},after row=\hline},%
string replace={NaN}{},%
every first column/.style={column type/.add={|}{}}]{runGenTable.out.txt}
\end{center}
\end{table}

The following lemma is evident and can be easily seen from Table~\ref{tab:herm}.
\begin{lemma}\label{lem:hermite_coefs}
For any $t \in \bbZp$ and $j \in \bbZp$,
\begin{enumerate}
\item $\rmH_{[2t+1,j]} = 0$, and
\item $\sign(\rmH_{[2t,j]}) = (-1)^t$.
\end{enumerate}
\end{lemma}
In order to derive a convenient computational procedure for $\mCoefSig$, we need additional notation. For a $\nu \in \bbZp^\nvar$, we define the \textit{Hermite $\nu$-shift} of an infinite array $\rmC$ as
\begin{equation}\label{eq:nushift_def}
\begin{split}
& S_{\nu} (\rmC) := \rmB,\quad \mbox{where} \\
&\rmB_{\alpha} = 
\begin{cases}
\rmC_{\alpha-\nu} \rmH_{[\nu_1, \alpha_1-\nu_1]} \cdots  \rmH_{[\nu_\nvar, \alpha_\nvar-\nu_\nvar]}, & \alpha \ge \nu \\
0, & \alpha_k < \nu_k  \quad \mbox{for some} \; k. 
\end{cases}
\end{split}
\end{equation}
\begin{example}
Consider the moment array 
\begin{equation}\label{eq:momarr_ex}
\mCoef = 
\bmx
\mcoef_{[0,0]} & \mcoef_{[0,1]} & \mcoef_{[0,2]} & \mcoef_{[0,3]} & \mcoef_{[0,4]} & \cdots\\
\mcoef_{[1,0]} & \mcoef_{[1,1]} & \mcoef_{[1,2]} & \mcoef_{[1,3]} & \cdots    &\\
\mcoef_{[2,0]} & \mcoef_{[2,1]} & \mcoef_{[2,2]} & \cdots    &           &\\
\mcoef_{[3,0]} & \mcoef_{[3,1]} & \vdots    &           &           &\\
\mcoef_{[4,0]} & \vdots    &           &           &           &\\
\vdots    &           &           &           &           &\\
\emx.
\end{equation}
(Only elements in $\trgset{2}{4}$ are shown.) Then its Hermite $\nu$-shift, for $\nu = [0,2]$, is
\[
S_{[0,2]} (\mCoef) = 
-
\bmx
0         & 0         & \mcoef_{[0,0]} & 3\mcoef_{[0,1]} & 6\mcoef_{[0,2]} & \cdots\\
0         & 0         & \mcoef_{[1,0]} & 3\mcoef_{[1,1]} & \cdots    &\\
0         & 0         & \mcoef_{[2,0]} & \cdots          &           &\\
0         & 0         & \vdots         &                 &           &\\
0         & \vdots    &                &                 &           &\\
\vdots    &           &                &                 &           &\\
\emx.
\]
\end{example}
The following property of $\nu$-shift immediately follows from Lemma~\ref{lem:hermite_coefs}.
\begin{corollary}\label{cor:nushift_prop}
If at least one element of $\nu$ is odd, then $S_{\nu} (\mCoef) = 0$.
\end{corollary}

\subsection{Construction of shifted moment arrays}
With the help of the introduced notation, the following theorem holds true.
\begin{theorem}\label{thm:adjusted_moments}
The $\sigma$-adjusted moment arrays $\mCoefSig$ can be computed using Hermite $\nu$-shifts as follows
\begin{equation}\label{eq:mcoefsig_poly}
\mCoefSig = \mCoefBas{0} + \sigma^2 \mCoefBas{1}  + \sigma^4 \mCoefBas{2} + \cdots,
\end{equation}
where $\mCoefBas{k}$ are \textit{basis arrays} for $\mCoefSig$, defined as
\begin{equation}\label{eq:mcoefbas_def}
\mCoefBas{k} := \sum_{\beta \in \degset{s}{k} \times \{(0, \ldots,0)\}} S_{2\beta} (\mCoef).
\end{equation}
In particular, $\mCoefBas{0} = \mCoef$.
\end{theorem}
\begin{proof}
Denote $\alpha' = \bmx \alpha_1 & \cdots & \alpha_s & 0 &\cdots &0\emx^{\top}$.
From \eqref{eq:mcoefsig_def} we have that 
\[
\begin{split}
\mcoefSig_\alpha &= 
 \sum\limits_{k=1}^{\npoints}
\prod_{j=1}^{s}\left(\sum\limits_{\nu_j=0}^{\alpha_j} \rmH_{[\nu_j,\alpha_j - \nu_j]} 
\sigma^{\nu_j} (\dpoint^{(k)}_j)^{\alpha_j -\nu_j}\right) 
\prod_{j=s+1}^{\nvar} (\dpoint^{(k)}_j)^{\alpha_j} \\
& = \sum\limits_{k=1}^{\npoints}\sum\limits_{\nu \in \boxset{\alpha'}}
  \sigma^{|\nu|} \prod_{j=1}^{\nvar} \rmH_{[\nu_j,\alpha_j - \nu_j]} (\dpoint^{(k)}_j)^{\alpha_j - \nu_j} \\
& = \sum\limits_{\nu \in \boxset{\alpha'}}
  \sigma^{|\nu|} \left(\prod_{j=1}^{\nvar} \rmH_{[\nu_j,\alpha_j - \nu_j]}\right)
  \left(\sum\limits_{k=1}^{\npoints} (\dpoint^{(k)})^{\alpha - \nu}\right) \\
& = \sum\limits_{\nu \in \boxset{\alpha'}} \sigma^{|\nu|} \left(S_{\nu} (\rmM)\right)_{\alpha}= \sum\limits_{\beta \in \bbZp^{s}\times \{(0, \ldots,0)\}} \sigma^{2|\beta|} \left(S_{2\beta} (\rmM)\right)_{\alpha},
\end{split}
\]
where the last two equalities follow from \eqref{eq:mcoefsig_def}, \eqref{eq:nushift_def} and Corollary~\ref{cor:nushift_prop}. This completes the proof.
\end{proof}
Note that from \eqref{eq:nushift_def}, for any $\alpha$, the coefficient $(\mCoefBas{k})_{\alpha}$ is equal to zero for all large enough $k$. Therefore the sum \eqref{eq:mcoefsig_poly} is element-wise finite and the definition \eqref{eq:mcoefsig_poly} is correct. In addition, the matrices $\Psi_k(\dpoints)$ defined in \eqref{eq:mpoly} are
\[
\Psi_k(\dpoints) = \scrH_{\calA} (\mCoefBas{k}).
\]

\begin{example}
Consider the case $s=\nvar=2$, and the moment array \eqref{eq:momarr_ex}.
We show only the elements in $\trgset{2}{4} = \trgset{2}{2} + \trgset{2}{2}$. Then we have that
\begin{equation}\label{eq:momarr1_ex}
\mCoefBas{1} = 
-
\bmx
0         & 0         & \mcoef_{[0,0]} & 3\mcoef_{[0,1]} & 6\mcoef_{[0,2]} & \cdots\\
0         & 0         & \mcoef_{[1,0]} & 3\mcoef_{[1,1]} & \cdots    &\\
\mcoef_{[0,0]} & \mcoef_{[0,1]} & \mcoef_{[0,2]} + \mcoef_{[2,0]} & \cdots    &           &\\
3\mcoef_{[1,0]} & 3\mcoef_{[1,1]} & \vdots    &           &           &\\
6\mcoef_{[2,0]} & \vdots    &           &           &           &\\
\vdots    &           &           &           &           &\\
\emx,
\end{equation}
and
\begin{equation}\label{eq:momarr2_ex}
\mCoefBas{2} = 
\bmx
0          & 0         & 0           & 0         & 3\mcoef_{[0,0]} & \cdots\\
0          & 0         & 0           & 0         & \cdots    &\\
0          & 0         & \mcoef_{[0,0]}  & \cdots    &           &\\
0          & 0         & \vdots      &           &           &\\
3\mcoef_{[0,0]} & \vdots    &             &           &           &\\
\vdots     &           &             &           &           &\\
\emx.
\end{equation}
For $\calA$ defined in Example~\ref{ex:conic2}, only the elements shown in \eqref{eq:momarr_ex}, \eqref{eq:momarr1_ex} and \eqref{eq:momarr2_ex} will appear in the matrix $\psimatcor{\calA}{\dpoints}{\sigma}$. It is easy to see that this is exactly (up to duplication and scaling of columns and rows) the matrix constructed in \cite{Kukush.etal04CSDA-Consistent,Shklyar.etal07JoMA-conic}.
\end{example}

\subsection{Existence of solutions of the polynomial eigenvalue problem}
Here we prove the existence of solution of \eqref{eq:als2_equation} under weaker assumptions that in \cite{Shklyar09PhD-Consistency}.
More precisely, we do not require Assumption~\ref{ass:closedness}.

\begin{theorem}
Assume that $\Sigma_0 = J_s$ and $\calA$ contains at least one multidegree $\alpha^{(k)}$ such that $|\alpha^{(k)}|_s$  is odd. Then for any data set $\dpoints$ there exists a solution to \eqref{eq:als2_equation} (i.e., there exists $\widehat{\sigma} \ge 0$ such that $\lambda_{min}(\psimatcor{\calA}{\dpoints}{\widehat{\sigma}}) = 0$).
\end{theorem}
\textbf{Proof.} 
Let $k$ be such that $|\alpha^{(k)}|_s = \typdeg$, and $\typdeg$ is odd. (For convenience we denote $\alpha = \alpha^{(k)}$.) Take $\theta = e_k$ (unit vector with $k$-th nonzero element). From \eqref{eq:psimatcor_qhankel} and \eqref{eq:mcoefsig_poly}, we have that
\[
Q_{\text{als},\sigma}(\theta,\dpoints) = \theta^{\top}  \psimatcor{\calA}{\dpoints}{\sigma^2} \theta = 
(\mCoefSig )_{2\alpha}  = 
c_0 + c_1 \sigma^2 + \cdots + c_\typdeg \sigma^{2\typdeg}, 
\]
where $c_j = (\mCoefBas{j})_{2\alpha}$.

 Now let us find the leading coefficient $c_\typdeg$. Denote $\alpha' = \bmx \alpha_1 & \cdots & \alpha_s & 0 &\cdots &0\emx^{\top}$.
By \eqref{eq:nushift_def}, $(S_{2\beta} (\mCoef))_{2\alpha} = 0$ for any $\beta \in (\degset{s}{\ell} \times (0, \ldots,0)) \setminus \alpha'$. Therefore, from \eqref{eq:mcoefbas_def} we have that 
\[
c_\typdeg = \rmH_{[2\alpha_1, 0]} \cdots  \rmH_{[2\alpha_{s}, 0]}.
\]
By Lemma~\ref{lem:hermite_coefs}, $\sign(\rmH_{[2t, 0]}) = (-1)^t$. Therefore,
\[
\mathop{\mathrm{sign}}(c_d) = (-1)^{\ell} < 0.
\]
Thus, there exists $\sigma_0>0$ such that $Q_{\text{als},\sigma_0}(\theta,\dpoints) < 0$ and $\psimatcor{\calA}{\dpoints}{\sigma_0}$ is not positive semidefinite. Hence, there exists a principal minor of $\psimatcor{\calA}{\dpoints}{\sigma_0}$, such that its determinant is negative at $\sigma_0$. Since the determinant of any minor is a polynomial function of $\sigma$, there exists $\sigma_1$, $0 \le \sigma_1 < \sigma_0$ such that one of the minors of $\psimatcor{\calA}{\dpoints}{\sigma_1}$ is zero and all the minors of $\psimatcor{\calA}{\dpoints}{\sigma}$, for $0 \le\sigma\le\sigma_1$ are nonnegative. Thus, $\psimatcor{\calA}{\dpoints}{\sigma_1}$ is rank deficient and positive semidefinite, which completes the proof.
\hfill$\Box$

%% file: invariance.tex
\section{Invariance properties of the estimators}\label{sec:invariance}
In this section, we assume that $\Sigma_0 = I$, and $\phi = \phi_{\calA}$ is given as \eqref{eq:monomials_calA}.

\subsection{Affine transformations and summary of results}
An affine transformation in $\bbR^{\nvar}$ is
\begin{equation}\label{eq:affine_trans}
T(d) := Kd + h,
\end{equation}
where $K \in \bbR^{\nvar \times \nvar}$ is a nonsingular matrix and $h \in \bbR^{\nvar}$.
We consider the following basic transformations:
\begin{enumerate}
\item \textit{orthogonal transformation}: $h=0$, $K$ --- orthogonal matrix ($KK^{\top} = K^{\top}K = I$), which includes rotation and reflections;
\item \textit{translation}: $K = I$, $h \neq 0$; and
\item \textit{uniform scaling}: $h = 0$, $K = \rho I$.
\end{enumerate}
All compositions of these basic transformations comprise the class of affine similarity transformations.

In Table~\ref{tab:invariance}, we summarize the conditions on the set of monomials $\calA$ under which the estimators are invariant for any given data $\dpoints$. The rows in Table~\ref{tab:invariance} correspond to the basic transformations and the columns correspond to the estimators (including the weighted norm under consideration).

\begin{table}[!hbt]
\caption{Summary of invariance properties of the estimators}
\label{tab:invariance}
\begin{center}
\begin{tabular}{|c|c|c|c|}\hline
& $\widehat{\theta}_{\text{ols}}$          & $\widehat{\theta}_{\text{als},\sigma}$ & $\widehat{\theta}_{\text{als}}$ \\\cline{2-4}
& \multicolumn{2}{|c|}{Bombieri norm} & any norm \\\hline
Orthogonal transformation & \multicolumn{3}{|c|}{$\calA \sim \degset{\nvar}{\typdeg_1} \cup \cdots \cup \degset{\nvar}{\typdeg_M}$ (Theorem~\ref{thm:rotation_invariance})}\\[1.5ex]\hline
Uniform scaling           & --- & --- &any $\calA$ (Theorem~\ref{thm:scaling_translation_invariance}) \\[1.5ex]\hline
Translation               & --- & --- &$\calA \sim \trgset{\nvar}{\typdeg}$ (Theorem~\ref{thm:scaling_translation_invariance}) \\[1.5ex]\hline
\end{tabular}
\end{center}
\end{table}

Most of the results are proved for the Bombieri norm.
\begin{definition}
The \emph{Bombieri norm} $\|\cdot\|^2_B$ is defined as
\begin{equation}\label{eq:bombieri}
\|\theta\|^2_B = \sum_{j=1}^{\nmind} \frac{\alpha^{(j)}_1! \cdots  \alpha^{(j)}_\nvar!}{(\alpha^{(j)}_1 + \cdots +\alpha^{(j)}_{\nvar})!} \theta_j^2,
\end{equation}
i.e., the coefficients are normalized by a multinomial coefficient.
\end{definition}
The Bombieri norm has the advantage that it is rotation-invariant. It is important to use the  Bombieri norm (and not just 2-norm, as in \cite{Markovsky12-Low}), in order to have rotation-invariant $\widehat{\theta}_{\text{ols}}$ and $\widehat{\theta}_{\text{als},\sigma}$ estimators.
\begin{example}
In Example~\ref{ex:conic2}, the Bombieri norm of the parameter vector is equal to
\[
\|\theta\|^2_B = \|A\|^2_F+ \|b\|^2_2 + c^2,
\]
where $\|\cdot\|_F$ is the Frobenius norm and $(A,b,c)$ are classic parameters for conic sections, i.e., $A\in \bbR^{2\times 2}$ (symmetric), $b\in \bbR^{2}$ and $c \in \bbR$ such that
\[
R_{\theta}(d) = d^{\top} A d +d^\top b + c.
\]
Thus the Bombieri norm coincides with the norm used in \cite{Shklyar.etal07JoMA-conic}.
\end{example}

\subsection{Some preliminary remarks}
Second, we note that the cost function $Q_{\text{als},\sigma}(\theta,\dpoints)$ defined in \eqref{eq:als_cost_def} can be expressed as a deconvolution of the cost function $Q_{\text{ols}}$.
\begin{equation}\label{eq:als_cost_def_deconv}
\begin{split}
Q_{\text{als},\sigma}(\theta,\dpoints) &= \theta^{\top} \psimatcor{\calA}{\dpoints}{\sigma} \theta = \sum_{k=1}^{\npoints} \sum\limits_{i,j=1}^{\nmind,\nmind} \theta_i \theta_j (\phi_i\phi_{j})\ast p_{-\sigma^2 \Sigma_0}(\dpoint^{(k)})\\
&= \sum_{k=1}^{\npoints} \left(\sum\limits_{j=1}^{\nmind} \theta_j \phi_{j}\right)^2\ast p_{-\sigma^2 \Sigma_0}(\dpoint^{(k)}) = \sum_{k=1}^{\npoints} R^2_{\theta} \ast p_{-\sigma^2 \Sigma_0}(\dpoint^{(k)}).
\end{split}
\end{equation}
In particular, the cost function \eqref{eq:als_cost_def_deconv} has the following property
\begin{equation}\label{eq:als_cost_def2}
\bfE (Q_{\text{als},\sigma}(\theta,\dpoints)) = Q_{\text{ols}}(\theta,\overline{\dpoints}).
\end{equation}

Second, we rewrite the \eqref{eq:als2_equation} using $Q_{\text{als},\sigma}(\theta,\dpoints)$.
The pair $(\widehat{\theta}_{\text{als}},\widehat{\sigma})$ is the solution of the following system of equations
\begin{equation}\label{eq:als2_opt_problem2}
\begin{split}
& Q_{\text{als},\widehat{\sigma}}(\widehat{\theta},\dpoints) = 0,\quad \|\widehat{\theta} \| = 1  \\
& Q_{\text{als},\widehat{\sigma}}(b,\dpoints) \ge 0 \mbox{ for all } b \neq 0.
\end{split}
\end{equation}

\subsection{Formal definition of invariance}
The estimation problems \eqref{eq:ols_opt_problem}, \eqref{eq:als1_opt_problem} and \eqref{eq:als2_opt_problem2} may have non-unique solutions. In order to handle this property, we introduce additional notation following \cite{Shklyar.etal07JoMA-conic}.
Let us fix an estimation problem and denote by
\[
\mathop{\mathrm{Sol}}(\dpoints) := \mbox{ set of solutions $\widehat{\theta}$ of the problem for a given }\dpoints.
\]
Then we can introduce a formal definition of invariance of a problem.
\begin{definition}[See {\cite[Definition 25]{Shklyar.etal07JoMA-conic}}]\label{def:invariance}
For a given set of points $\dpoints$, the estimation problem is called
\begin{itemize}
\item $T\Rightarrow$\emph{invariant}, if for all $\theta_1 \in 
\mathop{\mathrm{Sol}}(\dpoints)$ there exists $\theta_2 \in 
\mathop{\mathrm{Sol}}(T(\dpoints))$ such that
\[
R_{\theta_1}(d) = 0 \iff R_{\theta_2}(T(d)) = 0;
\]
\item $T\Leftarrow$\emph{invariant}, if for all $\theta_2 \in 
\mathop{\mathrm{Sol}}(T(\dpoints))$ there exists $\theta_1 \in 
\mathop{\mathrm{Sol}}(\dpoints)$ such that
\[
R_{\theta_1}(d) = 0 \iff R_{\theta_2}(T(d)) = 0;
\]
\item $T$-\emph{invariant} if it is both $T\Rightarrow${invariant} and $T\Leftarrow${invariant}.
\end{itemize}
\end{definition}

Obviously, an estimation problem which is invariant with respect to two transformations $T_1$ and  $T_2$, is also invariant to their composition $T_2 \circ T_1$.
\begin{note}
If $T_1$ and $T_2$ are two transformations such that
\begin{itemize}
\item for data $\dpoints$ the problem is $T_1$-invariant, and
\item for data $T_1(\dpoints)$ the problem is is $T_2$-invariant, 
\end{itemize}
then the estimation problem is $T_2 \circ T_1$-invariant for data $\dpoints$.
\end{note}

\subsection{Rotation invariance}
\begin{theorem}\label{thm:rotation_invariance}
For the Bombieri norm \eqref{eq:bombieri} and $\calA$ of the form
\begin{equation}\label{eq:union_homog}
\calA \sim \degset{\nvar}{\typdeg_1} \cup \cdots \cup \degset{\nvar}{\typdeg_M},
\end{equation}
the problems  \eqref{eq:ols_opt_problem}, \eqref{eq:als1_opt_problem} and \eqref{eq:als2_opt_problem2} are $T$-invariant for any orthogonal transformation and any dataset $\dpoints$.
\end{theorem}
{\bf Proof. } We divide the proof in three steps
\begin{enumerate}
\item (Parameter transformation.) An affine transformation applied to the data points can be mapped to transformation of parameters. Since the set of $\calA$ has the form \eqref{eq:union_homog}, the polynomial $R_{\theta}(d)$ is a sum of homogeneous polynomials
\[
R_{\theta}(d) = R_{\theta,1}(d) + \cdots + R_{\theta,M}(d)
\]
of  degrees $\{\ell_1,\ldots, \ell_K\}$. A linear transformation $T(d) = Kd$ maps homogeneous polynomials to homogeneous polynomials, hence there exists a parameter transformation $\calI_{T}:\bbR^{m} \to\bbR^{m}$, such that 
\[
R_{\theta}(d) = R_{\calI_{T}(\theta)}(T(d)) 
\]
holds in polynomial sense. 

For the inverse linear transformation $T^{-1}$, we have that 
\[
\calI_{T^{-1}}\big(\calI_{T}(\theta)\big) = \theta \quad \mbox{for all}\quad \theta \in \bbR^{\nvar}.
\]
Since $\calI_{T}$ is linear, it is a bijection that maps $\bbR^{m}$ to itself. 

If $T$ is an orthogonal transformation, from the properties of the Bombieri norm \cite[\S 5.3.E.7]{Borwein.Erdelyi95-Polynomials}, we have that $\|\calI_{T}(\theta)\|^2_B = \|\theta\|^2_B$, i.e., the transformation $\calI_{T}$ preserves the Bombieri norm.

\item (Invariance of  $\widehat{\theta}_{\text{ols}}$.) By definition of $Q_{\text{ols}}$, for any $\dpoints$,  we have that
\[
Q_{\text{ols}}(\theta,\dpoints) = \sum\limits_{k=1}^{N} R^2_{\theta}(\dpoint^{(k)})
= \sum\limits_{k=1}^{N} R^2_{\calI_{T}(\theta)}(T(\dpoint^{(k)})) = 
Q_{\text{ols}}(\calI_{T}(\theta),T(\dpoints)).
\]
Therefore, we have that
\begin{equation}\label{eq:ols_transformation}
\min_{\|\theta_1\|_B = 1} Q_{\text{ols}}(\theta_1,\dpoints)  = \min_{\|\theta_2\|_B = 1}Q_{\text{ols}} (\theta_2,T(\dpoints)),
\end{equation}
where $\theta_2 = \calI_{T}(\theta_1)$ is an invertible change of variables. It is easy to see from \eqref{eq:ols_transformation} that the conditions of Definition~\ref{def:invariance} are met 
for the estimation problem \eqref{eq:ols_opt_problem}.

\item (Invariance of  $\widehat{\theta}_{\text{als},\sigma}$.)
By Lemma~\ref{lem:deconvolution_prop}, we have that for any $\dpoints$ 
\begin{equation}\label{eq:als_transformation}
\begin{split}
&Q_{\text{als},\sigma}(\theta,\dpoints) = \sum\limits_{k=1}^{N} ((R^2_{\calI_{T}(\theta)} \circ T) \ast p_{-\sigma^2 I})(\dpoint^{(k)})\\
&\quad= \sum\limits_{k=1}^{N} ((R^2_{\calI_{T}(\theta)} \ast p_{-\sigma^2 I})\circ T)(\dpoint^{(k)}) = 
Q_{\text{als},\sigma}(\calI_{T}(\theta),T(\dpoints)),
\end{split}
\end{equation}
where the last but one equality holds because $KK^{\top} = I$. Therefore, we have that
\[
\min_{\|\theta_1\|_B = 1} Q_{\text{als},\sigma}(\theta_1,\dpoints)  = \min_{\|\theta_2\|_B = 1}  Q_{\text{als},\sigma} (\theta_2,T(\dpoints)),
\]
where $\theta_2 = \calI_{T}(\theta_1)$ is an invertible change of variables. Thus, problem \eqref{eq:ols_opt_problem} is $T$-invariant.

\item (Invariance of $\widehat{\theta}_{\text{als}}$.)
For this proof, we use the formulation \eqref{eq:als2_opt_problem2}. From \eqref{eq:als_transformation}, we have that the invertible change of parameters $\calI_T$ combined with transformation of data does not change the value of $Q_{\text{als},\sigma}(\theta,\dpoints)$. Thus the problem \eqref{eq:als2_opt_problem2} is $T$-invariant.
\end{enumerate}
\hfill$\Box$

\subsection{Scaling and translation invariance}

\begin{theorem}\label{thm:scaling_translation_invariance}
\begin{enumerate}
\item  For any $\calA$, the problem \eqref{eq:als2_opt_problem2} is invariant with respect to uniform scaling.
\item If  $\calA$ is  of the form
\[
\calA \sim \trgset{\nvar}{\typdeg},
\]
then the problem \eqref{eq:als2_opt_problem2} is translation-invariant.
\end{enumerate}
\end{theorem}
{\bf Proof. }
\begin{enumerate}
\item (Scaling invariance.) In this case, we have that the linear transformation has the form $T(d) = (\rho I) d$. Then, similarly to \eqref{eq:als_scale_transformation} have that
\begin{equation}\label{eq:als_scale_transformation}
\begin{split}
Q_{\text{als},\sigma}(\theta,\dpoints) &= \sum\limits_{k=1}^{N} (R^2_{\theta} \ast p_{-\sigma^2 I})(\dpoint^{(k)})
= \sum\limits_{k=1}^{N} ((R^2_{\calI_{T}(\theta)} \circ T) \ast p_{-\sigma^2 I})(\dpoint^{(k)})\\
&= \sum\limits_{k=1}^{N} ((R^2_{\calI_{T}(\theta)} \ast p_{-\sigma^2\rho^2 I})\circ T)(\dpoint^{(k)}) = 
Q_{\text{als},\rho\sigma}(\calI_{T}(\theta),T(\dpoints)).
\end{split}
\end{equation}
We have that $\theta_2 = \calI_{T}(\theta_1)$ and $\sigma_2 = \rho\sigma_1$ is an invertible change of variables. 
Combined with transformation of data, the change of variables, does not change the value of $Q_{\text{als},\sigma}(\theta,\dpoints)$.
Therefore, the problem \eqref{eq:als2_opt_problem2} is $T$-invariant.

\item (Translation invariance.) In this case, the affine transformation is $T(d) = d+b$, where $b \in \bbR^{\nvar}$. Since $\calA = \trgset{\nvar}{\typdeg}$, the polynomial $R_{\theta}(d)$ can be viewed as a homogeneous polynomial $f$ of degree $\typdeg$ in homogeneous coordinates:
\[
R_{\theta}(d) = f\left(\bmx d\\1\emx\right).
\]
The affine transformation $T$ is a linear transformation in homogeneous coordinates, and we have that 
\[
R_{\theta}(d) = f\left(\bmx d\\1\emx\right) = f_2\left(\bmx T(d)\\1\emx\right) = 
R_{\calJ_{T}(\theta)}(d),
\]
where $\calJ_{T}: \bbR^{m} \to \bbR^{m}$. As in the proof of Theorem~\ref{thm:rotation_invariance}, we have that 
\[
\calJ_{T}(\calJ_{T^{-1}}(\theta)) = \theta.
\]
Since $\calJ_T$ is linear, it is a bijection from $\bbR^{m}$ to $\bbR^{m}$. Similarly to \eqref{eq:als_transformation} and \eqref{eq:als_scale_transformation}, have that
\begin{equation*}
\begin{split}
Q_{\text{als},\sigma}(\theta,\dpoints) &= \sum\limits_{k=1}^{N} (R^2_{\theta} \ast p_{-\sigma^2 I})(\dpoint^{(k)})
= \sum\limits_{k=1}^{N} ((R^2_{\calJ_{T}(\theta)} \circ T) \ast p_{-\sigma^2 I})(\dpoint^{(k)})\\
&= \sum\limits_{k=1}^{N} ((R^2_{\calJ_{T}(\theta)} \ast p_{-\sigma^2 I})\circ T)(\dpoint^{(k)}) = 
Q_{\text{als},\sigma}(\calJ_{T}(\theta),T(\dpoints)).
\end{split}
\end{equation*}
Hence, $\calJ_{T}$ is an invertible change of variables, which does not change the value of $Q_{\text{als},\sigma}(\theta,\dpoints)$ when combined with transformation of data. 
Thus, the problem \eqref{eq:als2_opt_problem2} is $T$-invariant.
\end{enumerate}
\hfill$\Box$

\begin{note}
Theorems~\ref{thm:rotation_invariance}~and~\ref{thm:scaling_translation_invariance}  generalize  Theorems 28, 30, and 31 of \cite{Shklyar.etal07JoMA-conic}.
\end{note}

%% file: applications.tex
\section{Numerical examples}\label{sec:numerical}
All the examples in this section are reproducible and available at \url{http://github.com/slra/als-fit}.
\subsection{Invariance of the estimators}
We consider the example ``Special data'' from \cite{Gander.etal94BNM-Least}. The dataset consist of $8$ points, which are given by
\begin{equation}
\dpoints = 
\left\{
\bmx 1\\7 \emx,
\bmx 2\\6 \emx,
\bmx 5\\8 \emx,
\bmx 7\\7 \emx,
\bmx 9\\5 \emx,
\bmx 3\\7 \emx,
\bmx 6\\2 \emx,
\bmx 8\\4 \emx
\right\}.
\end{equation}
Next, we consider two affine similarity transformations of the dataset
\[
\dpoints_1 = T_1(\dpoints),\quad \dpoints_2 = T_2(\dpoints),
\]
where the
\[
T_1(d) = d + \bmx-13\\-3\emx,\quad T_2(d) =  
1.5 \left(
\bmx
\cos(\beta) & \sin(-\beta)\\ 
\sin(\beta) & \cos(\beta)  
\emx d +
\bmx
6\\3
\emx
\right), \; \beta = -\frac{\pi}{4}.
\]
For each of the datasets we compute ${\widehat{\theta}}_{\text{ols}}$, ${\widehat{\theta}}_{\text{als},\sigma}$, and
${\widehat{\theta}}_{\text{als}}$, for the Bombieri norm. In Fig.~\ref{fig:similarity_invariance}, it is shown that only ${\widehat{\theta}}_{\text{als}}$ remains invariant under the transformations $T_1$ and $T_2$. This agrees with the results of Section~\ref{sec:invariance}, since both transformations contain a translation.
\begin{figure}[!ht]%
\centering%
\includegraphics[width=0.75\textwidth]{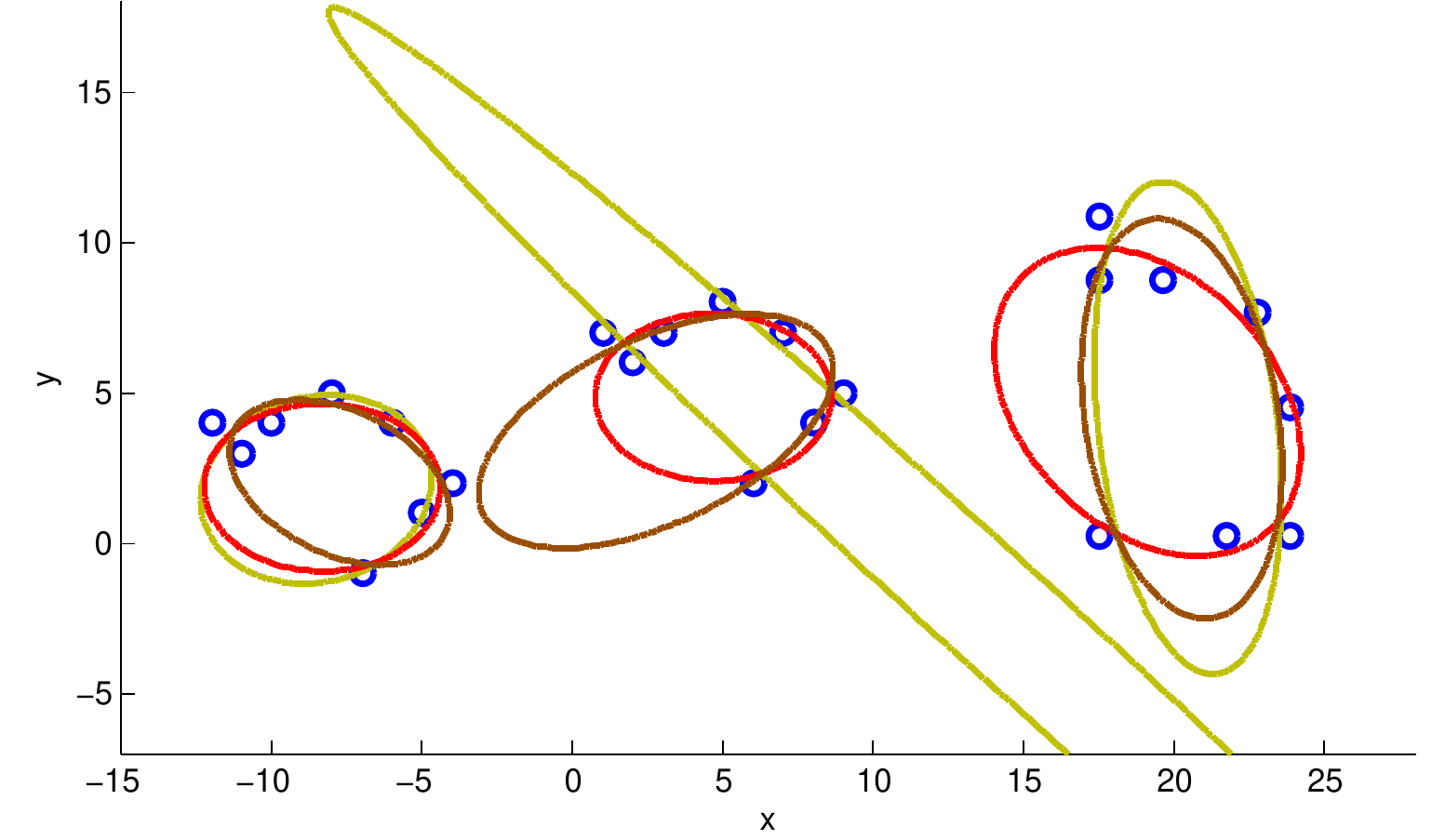}%
\caption{Fitting a conic section: blue circles --- data points ${d}^{(j)}$; green curve --- ${\widehat{\theta}}_{\text{ols}}$, brown curve --- ${\widehat{\theta}}_{\text{als},\sigma}$, red curve --- ${\widehat{\theta}}_{\text{als}}$.}%
\label{fig:similarity_invariance}
\end{figure}

Next, we demonstrate the importance of Bombieri norm for rotation invariance of ${\widehat{\theta}}_{\text{ols}}$ and ${\widehat{\theta}}_{\text{als},\sigma}$. We consider the dataset $\dpoints$ with coordinates given in Table~\ref{tab:dataset_rotinvar}.
\begin{table}[!hbt]
\caption{Test dataset for rotation invariance. First row: index of the point ($j\in\{0,\ldots, 12\}$). Second and third rows: coordinates of the points.}
\label{tab:dataset_rotinvar}
\begin{center}
{\everymath{\scriptstyle}  
\begin{filecontents*}{Dtestrot.txt}
0.4   1   1   1.2 1.7 1.8 1.7 2.2  1.2 1.5  0.3 0.3  0.3  
0.4   0.2 0.6 1.0 1.2 1.2 1.5 1.4  1.6 2.0  1.3 0.8  0.3 
\end{filecontents*}
\pgfplotstabletypeset[every head row/.style={before row={},after row=\hline},%
string replace={NaN}{},%
every first column/.style={column type/.add={|}{}}]{Dtestrot.txt}}
\end{center}
\end{table}

We also construct a transformed dataset $\dpoints_1$, which is $\dpoints$ rotated by $\frac{2\pi}{3}$ around the origin. Next, we fix $\calA \sim \trgset{2}{2}$, and calculate ${\widehat{\theta}}_{\text{ols}}$ for two different norms: Bombieri norm and the ordinary $2$-norm. In Fig.~\ref{fig:rotation_invariance}, the results of fit for two estimators are shown ($\widehat{\theta}_{\text{als}}$ is shown for reference). 
\begin{figure}[!ht]%
\centering%
\includegraphics[width=0.5\textwidth]{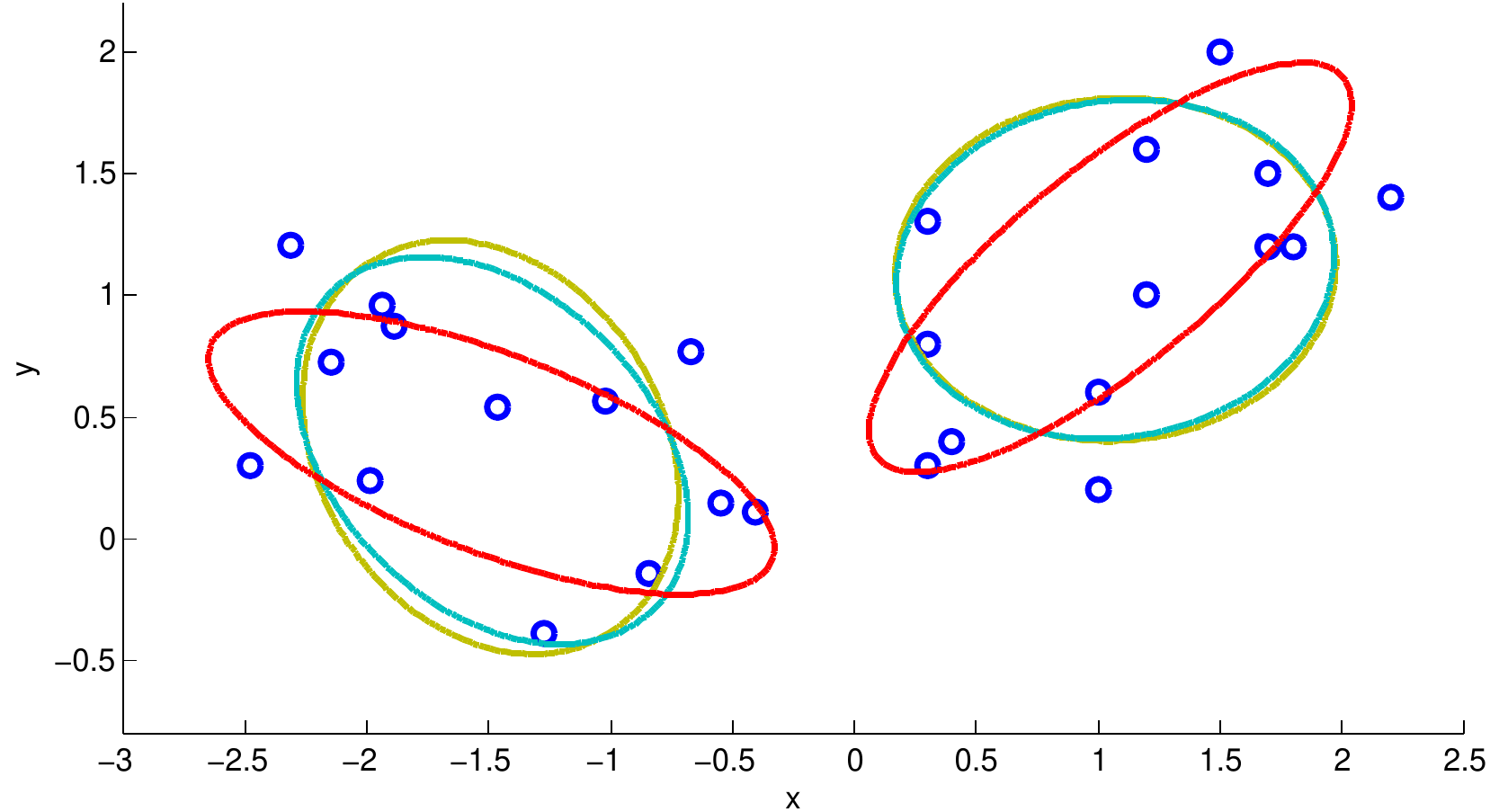}%
\caption{Fitting a conic section: blue circles --- data points ${d}^{(j)}$; green curve --- ${\widehat{\theta}}_{\text{ols}}$ (Bombieri norm), blue curve --- ${\widehat{\theta}}_{\text{ols}}$ ($2$-norm), red curve --- ${\widehat{\theta}}_{\text{als}}$.}%
\label{fig:rotation_invariance}
\end{figure}
The results in Fig~\ref{fig:rotation_invariance} show that ${\widehat{\theta}}_{\text{ols}}$ is invariant under rotation only if the Bombieri norm is used.

\subsection{Consistency of the estimators}
Next, we show the consistency of the estimators, proved in \cite{Shklyar09PhD-Consistency}.
For each $\npoints$, we define the set of true data points $\overline{\dpoints}_N$.
For each $j=1,\ldots,M$, we draw a realization of the noisy data points according to \eqref{eq:eiv}, and denote it by $\dpoints_{N,j}$. For an estimator $\widehat{\theta}$, we compute its value for the  $j$-th dataset as $\widehat{\theta}_{N,j}$, which allows us to estimate the spread of the estimator as
\[
s(\widehat{\theta}, N) := \frac{\sum\limits_{j=1}^M \sin^2 (\angle(\widehat{\theta}_{N,j}, \overline{\theta}))}{M} = \frac{\sum\limits_{j=1}^M 1 - \frac{\left(\overline{\theta}^{\top}\widehat{\theta}_{N,j}\right)^{2}}{\|\widehat{\theta}_{N,j}\|_2^2 \|\overline{\theta}\|_2^2}}{M}.
\]
The sum of squared sines is chosen because the estimates $\widehat{\theta}$ and $-\widehat{\theta}$ are equivalent (since the parameter is essentially defined on the projective space).

We consider an example of eight curve, which has an implicit representation
\[
x^4 = x^2-y^2,
\]
and a parametric representation
\[
\begin{split}
x(t) &= \sin(2\pi t), \\
y(t) &= \sin(2\pi t) \cos(2\pi t).
\end{split}
\]%
%
For each $N$, we define the set of true data points
as uniformly distributed in the parameter, i.e.
\[
\overline{\dpoints}_{\npoints} = \left\{\bmx x\left(0\right)\\ y\left(0\right)\emx, \bmx x\left(\frac{1}{N}\right)\\ y\left(\frac{1}{N}\right)\emx, \ldots, \bmx x\left(\frac{N-1}{N}\right)\\ y\left(\frac{N-1}{N}\right)\emx\right\}.
\]

A realization of noisy $\dpoints_{100,1}$ is shown in figure is plotted in Fig.~\ref{fig:eight}. Fig.~\ref{fig:eight} illustrates the meaning of consistency: all the points are noisy, but with increasing number of data points, the estimate approaches the true value. We see that even for small noise, the OLS estimator gives poor results. 

\begin{figure}[!ht]%
\centering%
\includegraphics[width=0.75\textwidth]{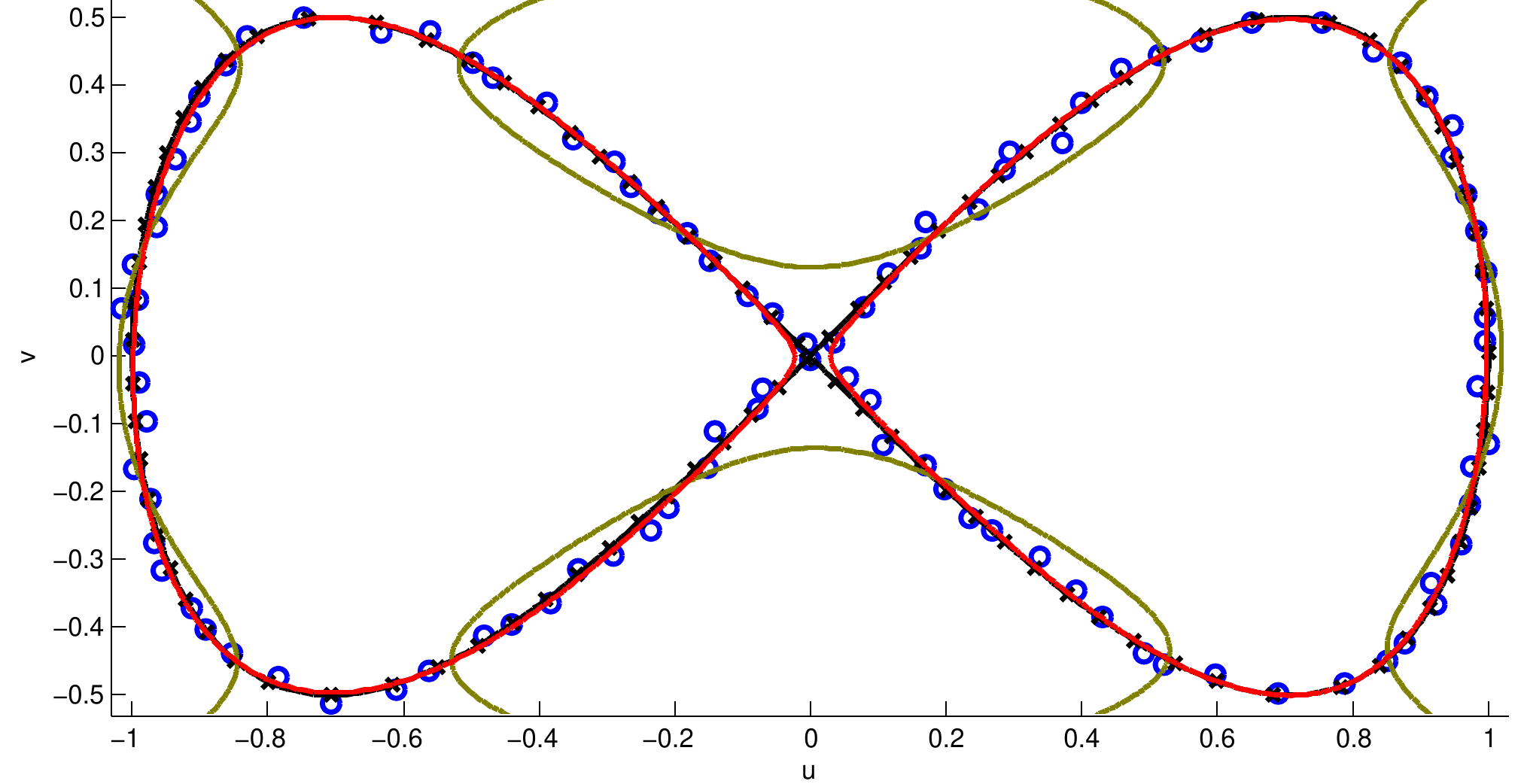}%
\caption{Eight curve: blue circles --- data points ${d}^{(j)}$; black curve --- true curve, black crosses --- true data points $\overline{d}^{(j)}$, green curve --- ${\widehat{\theta}}_{\text{ols}}$, red curve --- ${\widehat{\theta}}_{\text{als}}$.}%
\label{fig:eight}
\end{figure}

We fix the matrix of multidegrees as $\calA\sim \trgset{2}{4}$, consider number of data points as $N = 2^j, j = 7, \ldots,17$, and set the number of realizations to $M = 100$. The noise standard deviation is  $\sigma = 0.01$.
In Fig.~\ref{fig:eight_consistency}, we plot the spread of the estimators $s(\cdot, N)$ depending on $N$.  We also consider two noise scenarios: Gaussian noise ($\widetilde{d} \sim \normal(0,\sigma^2 I)$) and uniform noise with the same variance of the coordinates ($\widetilde{d}_j$ are  uniformly distributed on $[-\sqrt{3}\sigma;\sqrt{3}\sigma]$). 

\begin{figure}[!ht]
\centering
\pgfplotsset{width=0.45\linewidth}
\begin{tikzpicture}[baseline]
\begin{loglogaxis}[
name=plot1,
legend columns=1, legend style={at={(0.1, 0.1)}, anchor=south west},
xtickten={1,2,3,4,5},
xlabel={N}, ylabel={$s(\cdot, N)$}
]
\addplot [no marks] table[x=Ns, y=ols]{eightNs_noise1.txt};
\addlegendentry{$\widehat{\theta}_{\text{ols}}$}
\addplot [red, no marks] table[x=Ns, y=als2]{eightNs_noise1.txt};
\addlegendentry{$\widehat{\theta}_{\text{als}}$}
\end{loglogaxis}
\hskip.5cm
\begin{loglogaxis}[
at=(plot1.right of south east), anchor = left of south west,
legend columns=1, legend style={at={(0.1, 0.1)}, anchor=south west},
xtickten={1,2,3,4,5},
xlabel={N}, ylabel={$s(\cdot, N)$}
]
\addplot [no marks] table[x=Ns, y=ols]{eightNs_noise2.txt};
\addlegendentry{$\widehat{\theta}_{\text{ols}}$}
\addplot [red, no marks] table[x=Ns, y=als2]{eightNs_noise2.txt};
\addlegendentry{$\widehat{\theta}_{\text{als}}$}
\end{loglogaxis}
\end{tikzpicture}
\caption[estCovScenario1]{RMSE of the errors of the estimators, eight curve. Left: Gaussian noise. Right: uniform noise.}
\label{fig:eight_consistency}
\end{figure}
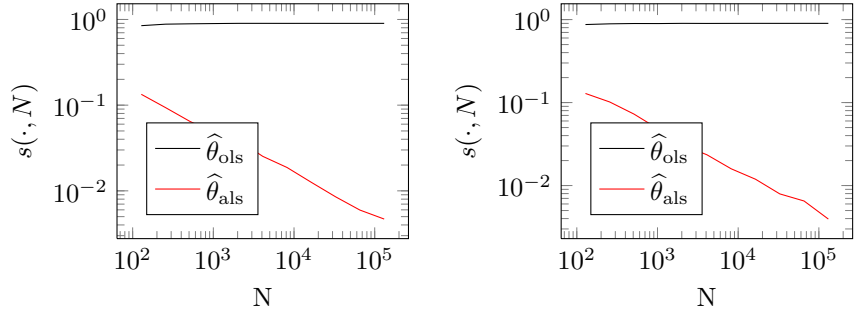

As shown in Fig.~\ref{fig:eight_consistency}, for the algebraic fit
the RMS error converges to a non-zero value, whereas for the ALS estimator, the RMS error converges to zero, as predicted by Theorem~\ref{thm:consistency}. Surprisingly, the convergence to $0$ also seems to take place for the wrong (uniform) noise model. Similar results are observed for the estimate of $\widehat{\sigma}^2$, for which the RMSE plots are shown in Fig.~\ref{fig:eight_consistency_sigmahat}.

\begin{figure}[!ht]
\centering
\pgfplotsset{width=0.45\linewidth}
\begin{tikzpicture}[baseline]
\begin{loglogaxis}[
name=plot1,
xtickten={1,2,3,4,5},
xlabel={N}, ylabel={$RMSE(\widehat{\sigma}^2)$}
]
\addplot [no marks] table[x=Ns, y=sigmahat]{eightNs_noise1.txt};
\end{loglogaxis}
\hskip.5cm
\begin{loglogaxis}[
at=(plot1.right of south east), anchor = left of south west,
xtickten={1,2,3,4,5},
xlabel={N}, ylabel={$RMSE(\widehat{\sigma}^2)$}
]
\addplot [no marks] table[x=Ns, y=sigmahat]{eightNs_noise2.txt};
\end{loglogaxis}
\end{tikzpicture}
\caption[estCovScenario1]{RMSE of the errors of $\widehat{\sigma}^2$, eight curve. Left: Gaussian noise. Right: uniform noise.}
\label{fig:eight_consistency_sigmahat}
\end{figure}

Finally, we study the behavior of the estimates as $\sigma$ varies. This is the setting which is often used in the literature on curve fitting \cite{Chernov10-Circular}.
 We fix $N=1000$ and choose $\sigma=10^{-6} \cdot 2^{j}$, $j=0,\ldots,13$, and plot relative error $\frac{s(\widehat{\theta}, N)}{\sigma}$ and scaled RMSE of $\sigma^2$, depending on $\sigma$ in Fig.~\ref{fig:eight_consistency_varsigma}. 
The added noise is uniform (the wrong noise model).

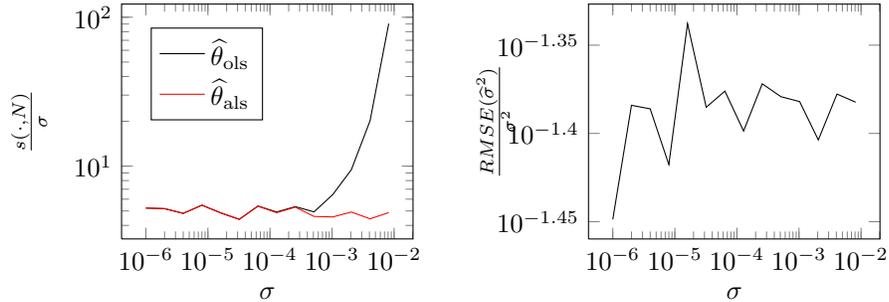
\begin{figure}[!ht]
\centering
\pgfplotsset{width=0.45\linewidth}
\begin{tikzpicture}[baseline]
\begin{loglogaxis}[
name=plot1,
legend columns=1, legend style={at={(0.1, 0.9)}, anchor=north west},
xtickten={-6,-5,-4,-3,-2},
xlabel={$\sigma$}, ylabel={$\frac{s(\cdot, N)}{\sigma}$}
]
\addplot [no marks] table[x=sigmas, y expr=\thisrowno{2}/\thisrowno{1}]{eightSigmas_noise2.txt};
\addlegendentry{$\widehat{\theta}_{\text{ols}}$}
\addplot [red, no marks] table[x=sigmas, y expr=\thisrowno{3}/\thisrowno{1}]{eightSigmas_noise2.txt};
\addlegendentry{$\widehat{\theta}_{\text{als}}$}
\end{loglogaxis}
\hskip.5cm
\begin{loglogaxis}[
at=(plot1.right of south east), anchor = left of south west,
xtickten={-6,-5,-4,-3,-2},
xlabel={$\sigma$}, ylabel={$\frac{RMSE(\widehat{\sigma}^2)}{\sigma^2}$}
]
\addplot [no marks] table[x=sigmas, y expr=\thisrowno{5}/(\thisrowno{1}*\thisrowno{1})]{eightSigmas_noise2.txt};
\end{loglogaxis}
\end{tikzpicture}
\caption[estCovScenario1]{Relative RMSE of the estimators depending on $\sigma$, eight curve. Left: ${\theta}$. Right: $\sigma^2$.}
\label{fig:eight_consistency_varsigma}
\end{figure}

In Fig.~\ref{fig:eight_consistency_varsigma}, we see that $\widehat{\theta}_{\text{als}}$ behaves better for higher values of noise, and preserves ratio of the magnitude output error to the magnitude of the input error (which can be interpreted as the condition number of the problem).

\subsubsection{Subspace clustering}
Next, we consider an example with higher dimensions ($\nvar = 3$), and also when the conditions of Theorem~\ref{thm:consistency} are not satisfied. The example is a union of three hyperplanes, which is inspired by an application in subspace clustering \cite{Vidal11ISPM-Subspace}.

Let $b^{(1)},\ldots, b^{(r)} \in \bbR^{\nvar}$ be a pairwise non-collinear vectors, and
\begin{equation}\label{eq:ssclust_model}
\left((b^{(1)})^{\top} d = 0 \right) \vee \cdots\vee
\left((b^{(r)})^{\top} d = 0 \right),
\end{equation}
be a union of hyperplanes, for which the normal vectors are $b^{(j)}$. Then the set of solutions of  \eqref{eq:ssclust_model} is an algebraic hypersurface, since \eqref{eq:ssclust_model} is equivalent to
\begin{equation}\label{eq:factorizability}
\left((b^{(1)})^{\top} d \right) \cdot \cdots\cdot
\left((b^{(r)})^{\top} d  \right) = 0.
\end{equation}
The set of monomials in \eqref{eq:factorizability} is $\degset{\nvar}{r}$. As noted in \cite{Vidal11ISPM-Subspace}, modeling the data as a union of hyperplanes may be posed as an algebraic hypersurface fitting problem. Typically, algebraic fitting (i.e., $\widehat{\theta}_{\text{ols}}$) is used for this purpose $\widehat{\theta}$. In what follows, we show that the ALS fitting should be preferred.

We consider the following three  vectors
\[
b^{(1)} = \bmx0&1&0\emx^{\top},\quad b^{(2)} = \bmx\sqrt{2}&\sqrt{2}&0\emx^{\top}, \quad b^{(2)} = \bmx\sqrt{3}&\sqrt{3}&\sqrt{3}\emx^{\top}.
\]
We fix the noise standard deviation to $\sigma = 0.05$, and generate the true points as follows. We randomly assign points to the hyperplanes (with equal probability). In each hyperplane, the true points are distributed uniformly in a $1\times 1$ square. An example of noisy data points is shown in Fig.~\ref{fig:subspace}.
\begin{figure}[!ht]%
\centering%
\includegraphics[width=0.5\textwidth]{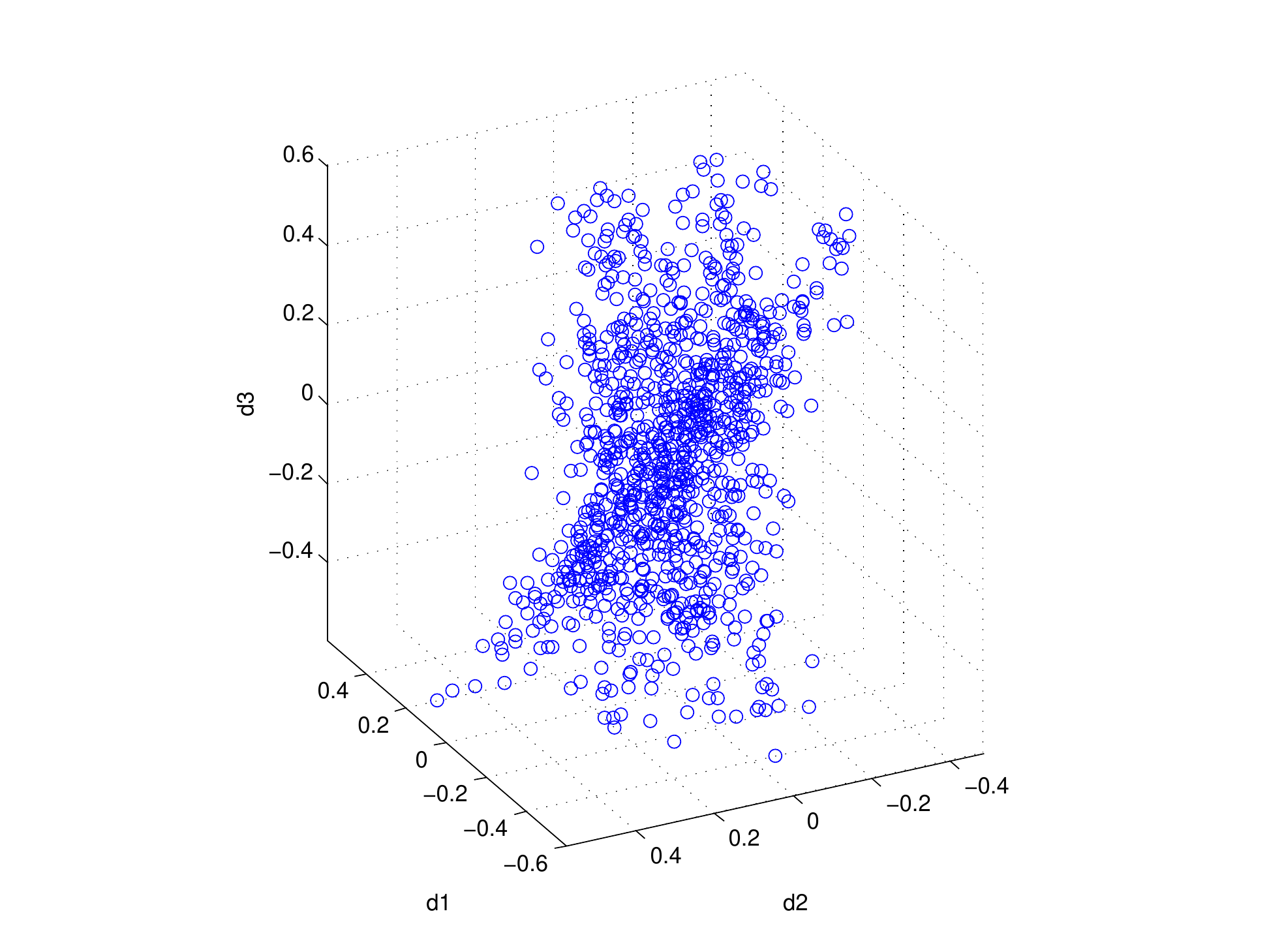}%
\caption{Noisy points on a union of three hyperplanes.}%
\label{fig:subspace}
\end{figure}
Next, we choose the matrix of multidegrees as $\calA\sim \degset{3}{3}$, consider number of data points as $N = 2^j, j = 7, \ldots,17$, and set the number of realizations to $M = 100$. The noise standard deviation is  $\sigma = 0.01$.
In Fig.~\ref{fig:subspace_consistency}, we plot the RMSE of the estimators $s(\cdot, N)$. 

\begin{figure}[!ht]
\centering
\pgfplotsset{width=0.45\linewidth}
\begin{tikzpicture}[baseline]
\begin{loglogaxis}[
name=plot1,
legend columns=1, legend style={at={(0.1, 0.1)}, anchor=south west},
xtickten={1,2,3,4,5},
xlabel={N}, ylabel={$s(\cdot, N)$}
]
\addplot [no marks] table[x=Ns, y=ols]{subspaceNs_noise1.txt};
\addlegendentry{$\widehat{\theta}_{\text{ols}}$}
\addplot [red, no marks] table[x=Ns, y=als2]{subspaceNs_noise1.txt};
\addlegendentry{$\widehat{\theta}_{\text{als}}$}
\end{loglogaxis}
\hskip.5cm
\begin{loglogaxis}[
at=(plot1.right of south east), anchor = left of south west,
legend columns=1, legend style={at={(0.1, 0.1)}, anchor=south west},
xtickten={1,2,3,4,5},
xlabel={N}, ylabel={$s(\cdot, N)$}
]
\addplot [no marks] table[x=Ns, y=ols]{subspaceNs_noise2.txt};
\addlegendentry{$\widehat{\theta}_{\text{ols}}$}
\addplot [red, no marks] table[x=Ns, y=als2]{subspaceNs_noise2.txt};
\addlegendentry{$\widehat{\theta}_{\text{als}}$}
\end{loglogaxis}
\end{tikzpicture}
\caption[estCovScenario1]{RMSE of the errors of the estimators, subspace clustering. Left: Gaussian noise. Right: uniform noise.}
\label{fig:subspace_consistency}
\end{figure}
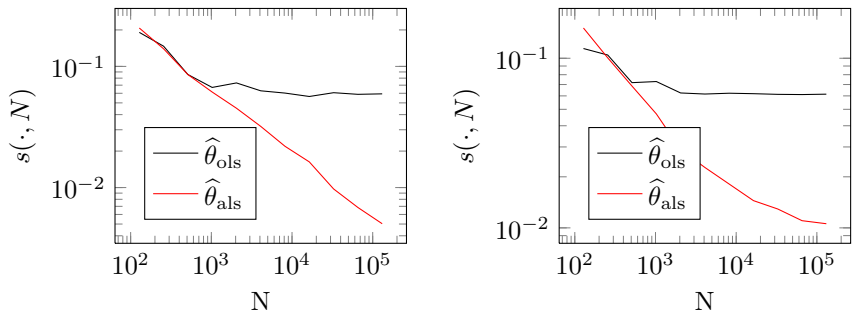

As shown in Fig.~\ref{fig:subspace_consistency}, the OLS estimator is again biased, and the ALS estimator seems to converge to zero for the correct noise model. For the wrong noise model (uniform noise), the estimator seems to be inconsistent, but has a smaller asymptotic bias. 
We note that for small $N$, the OLS estimator is slightly better than the ALS estimator. 
However, for large $N$ the ALS estimator clearly outperforms the algebraic fitting.
Note that the conditions of Theorem~\ref{thm:consistency} are not satisfied, since the set of polynomials $\phi_{\calA}$ does not satisfy Assumption~\ref{ass:closedness}.

\section{Conclusions}
In this paper, we considered the adjusted least squares estimators (in the cases of known and unknown variance) for algebraic hypersurfaces with arbitrary support. We  showed that the matrix coefficients of the matrix polynomial can be constructed  as quasi-Hankel matrices from shifts of the moment array. This allowed us to prove a new sufficient condition for existence of the ALS estimator. We also derived conditions for rotation/scaling/translation invariance of the estimators, and showed that in many cases it is important to use the Bombieri norm. Finally, we demonstrated on numerical experiments that the ALS estimator works well beyond its probabilistic model and known results on its consistency. We believe that the ALS estimator can be used  as a general-purpose hypersurface fitting tool, and that its properties deserve further theoretical and numerical investigation.